\documentclass[12pt,titlepage, authoryear]{article}

\usepackage{amssymb}

\usepackage{amsfonts,amsthm,graphicx,threeparttable,latexsym}
\usepackage[title]{appendix}
\RequirePackage[OT1]{fontenc}
\RequirePackage{graphicx}
\RequirePackage{amsthm}
\RequirePackage[cmex10]{amsmath}
\RequirePackage{natbib}
\RequirePackage[colorlinks,citecolor=blue,urlcolor=blue]{hyperref}
\RequirePackage{amssymb}
\theoremstyle{plain}

\newtheorem{proposition}{Proposition}[section]
\newtheorem{lemma}{Lemma}[section]
\newtheorem{remark}{Remark}[section]

\providecommand{\norm}[1]{\lVert#1\rVert}
\begin{document}




\title{Semi-Global Solutions to DSGE Models:
Perturbation around a Deterministic Path}


\author{Viktors Ajevskis\footnote{email: Viktors.Ajevskis@bank.lv}\\Bank of Latvia\footnote{The views  expressed in this paper are the sole responsibility of the author and do not
necessarily reflect the position of the Bank of Latvia.}}
\maketitle

\begin{abstract}
This study proposes an approach based on a perturbation technique to construct global solutions to dynamic stochastic general equilibrium models (DSGE). The main idea is to expand a solution in a series of powers of a small parameter scaling the uncertainty in the economy around a solution to the deterministic model, i.e. the model where the volatility of the shocks vanishes. If a deterministic path is global in state variables, then so are the constructed solutions to the stochastic model, whereas these solutions are local in the scaling parameter. Under the assumption that a deterministic path is already known the higher order terms in the expansion are obtained recursively by solving linear rational expectations models with time-varying parameters. The present work also proposes a method rested on backward recursion for solving  general systems of linear rational expectations models with
time-varying parameters and determines the conditions under which the solutions of the method exist. 
\end{abstract}

\section{Introduction}
\label{intro}

%
Perturbation methods applied in macroeconomics are used to expand the exact solution around a deterministic steady state in powers of state variables
and a small
parameter scaling the uncertainty in the economy. 
The solutions based on the Taylor series expansion are intrinsically local, i.e. they are accurate in some neighborhood (presumably small) of the deterministic steady state. 
Out of the neighborhood the solutions may behave odd, for example,  can imply
explosive dynamics  \citep{r10}. 
The other problem with the perturbation method is that we do not know a priori for non-trivial models how small the neighborhood must be to achieve a given level of accuracy. 

The recent crisis has renewed an interest in methods that provide global solutions to DSGE models, i.e. the solutions some points of which are far away from the steady state. This may occur after a big shock hitting the economy, or if the initial conditions are far away from the steady state, the examples of this situation are  the economies in transition and developing economies. 

This study presents an approach based on a perturbation technique to construct global solutions to DSGE models. The proposed solutions are represented as a series in powers of a small parameter $\sigma$ scaling the covariance matrix of the shocks.
The zero order approximation corresponds to the solution to the deterministic model, because all shocks vanish as $\sigma =0$. Global solutions to deterministic models can be obtained reasonably fast by effective numerical methods\footnote{The algorithms incorporated in the widely-used software such as Dynare (and less available Troll) find a stacked-time solution and are based on Newton's method combined with sparse-matrix techniques \citep{r0}.}
even for large size models \citep{r36}. 
 For this reason the next stages of the method are implemented  assuming that the solution to the deterministic model under  given initial conditions is known.

Higher-order systems  depend only on quantities of lower orders, hence they can be solved recursively. The homogeneous part of these systems  is the same for all orders and
depends on the deterministic solution.  Consequently, each system can be represented as a rational expectation model with
time-varying parameters.
In the case of  rational expectations models with constant parameters  the stable block of equations can be isolated and solved forward. This is not possible for models with time-varying parameters. 

The other contribution of the present work is a  method proposed for solving general systems of linear rational expectations models with
time-varying parameters and determines the conditions under which the solutions of the method exist. 
%
The method starts with finding a finite-horizon solution by using backward recursion. Next we prove that as the horizon tends to infinity the finite-horizon solutions approach to a limit solution that is bounded for all positive time. 
The proposed method for solving linear rational expectations models with
time-varying parameters may be valuable in itself, for example,
 for solving models with anticipated structural and policy changes. 
%
%



%
Notice that whenever the deterministic solution is global in state variables so is the approximate solution to the stochastic problem. At the same time, if the parameter $\sigma$ is small enough, then the solution obtained is close to the deterministic one.
For this reason, we shall call this approach semi-global.

To illustrating how the method works 
we apply it to the asset pricing model of \cite{r16}. 
The simplicity of the model allows for obtaining the approximations in an analytical form.
We compare the policy functions  of the
second order solution of the semi-global method with the local
Taylor series expansion of orders two and six \citep{r14}. The results show that the semi-global solution is more accurate, in some sense, than even the sixth order of the local Taylor expansion.

This paper contributes to a growing literature on using the
perturbation technique for solving DSGE models. The perturbation methodology
in economics has been advanced by Judd and co-authors as in \cite{r8,r17,r18}. \cite{r7} give a theoretical basis for using perturbation
methods in DSGE modeling; namely, applying the implicit function theorem, they
prove that the perturbed rational expectations solution continuously depends
on a parameter and therefore tends to the deterministic solution as the
parameter tends to zero.  

Almost all of the literature is concerned with the approximations
around the steady state as in \cite{r3,r14,r10,r5}. 
\cite{r19} and \cite{r19_} make use of series expansion in powers of $\sigma $ 
to provide a theoretical foundation for pruning methods \citep{r10}, which is aimed to avoid the explosive behavior of a solution. \cite{r1} develop the same approach for higher-order approximations. 
Lombardo and Uhlig's approach can be treated as a special case of the method proposed in the current
study, namely a deterministic solution around which the expansion is used
 is only the steady state.  
Both approaches based on the perturbation methodology used in applied mathematics ( \cite{r88} and \cite{r20}).
The essence of the methodology  is to expand a solution in a series of powers of a small parameter, and thus obtain a set of problems that can be solved recursively. It is supposed that each of this problems is  easier to solve than the original one. Actually  in applied mathematics literature (\cite{r88} and \cite{r20}), the zeroth-order approximation is typically a function of time $t$ rather than a steady state as in \cite{r19}, \cite{r19_} and \cite{r1}. 
 \citet[Chapter 13]{r8} outlines how to apply  perturbations around
the known entire solution, which is not necessarily the steady state. He
considers a simple continuous-time stochastic growth models
in the dynamic programming framework. This paper develops an 
 approach to construct approximate solutions to  discrete-time DSGE models in  general form by using the
perturbation method around a global deterministic path.

Despite the fact that the pruning procedure avoids the explosive behavior of a solution, it remains local, and as such may have some undesirable properties. For example,  the pruning procedure might 
provide a first few impulse responses with wrong signs under a sufficiently large shock.  This case seems even worse than the explosive dynamics since the impulse responses for a first few periods are most interesting and relevant for theoretical implications of a model as well as a policy analysis; therefore, their incorrect signs could mislead a researcher or a policymaker. In this situation the pruning procedure just conceals the real problem.
As we will show in the example, the problem with a wrong sign of impulse responses can occur even in a situation where the pruning is not needed.


%

The rest of the paper is organized  as follows. The next section presents the model set-up. Section~\ref{expan} provides a detailed exposition of series expansions for DSGE models. In Section~\ref{transform} we transform  the model into a convenient form to deal with. Section~\ref{sec6} presents the method for solving  rational expectations models for time-varying parameters. 
The proposed method is applied to an asset pricing model in Section~\ref{example}, where it is also compared with the local Taylor series expansions. 
Conclusions are presented in Section~\ref{conclusion}.

\section{The Model}
\label{sec:2}
DSGE models usually have the form
\begin{eqnarray}
&E_{t}f(y_{t+1},y_{t},x_{t+1},x_{t},z_{t+1},z_{t})  =  0,\label{1}\\
&z_{t+1}  =  \Lambda z_{t} + \sigma\varepsilon _{t+1}, \label{2} 
\end{eqnarray}
where $E_{t}$ denotes the conditional expectations operator,  $x_t$ is an $n_x\times 1$ vector containing the $t$-period endogenous  state variables; 
$y_t$ is an $n_y\times 1$ vector containing the $t$-period endogenous variables that are not state variables; $z_t$ is an $n_z\times 1$ vector containing the $t$-period 
exogenous  state  variables; 
$\varepsilon _{t}$  is the vector with the corresponding innovations;  and the $ n_z\times n_z$ covariance matrix $\Omega$; 
$f$ maps ${\mathbb{R}}^{n_y}\times{\mathbb{R}}^{n_y}\times{\mathbb{R}}^{n_x}\times{\mathbb{R}}^{n_x}\times{\mathbb{R}}^{n_{z}}\times{\mathbb{R}}^{n_{z}}$   into ${\mathbb{R}}^{n_y}\times{\mathbb{R}}^{n_x}$ and is assumed to be sufficiently smooth. The scalar $\sigma$ ($\sigma>0$) is a scaling parameter for the disturbance terms $\varepsilon _{t}$. We assume that all mixed moments of $\varepsilon _{t}$ are finite. All eigenvalues of the matrix $\Lambda$ have modulus less than one.  
%
The problem 
is to find a stable solution $(x_{t}, y_{t})$ to \eqref{1} 
for a given initial condition $(x_{0}, z_{0})$. 
A process is stable if its unconditional expectations are bounded \citep{r5_}.

\section{Series Expansion}\label{expan}
\subsection{The General Case}
In this section we shall follow the perturbation methodology used in applied mathematics (see,
for example, \cite{r88} and \cite{r20}) to derive an approximate solution to the model %
\eqref{1}--\eqref{2}. 
For small $\sigma $, we assume that the
solution\footnote{As is conventional in applied mathematics literature (see for example \cite{r20}), each term of the expansion is not divided by the factorial term at this stage.}
has a particular form of expansions 
\begin{equation}  \label{5}
y_{t} = y^{(0)}_t +  \sigma y^{(1)}_t +\sigma ^{2} y^{(2)}_t + \cdots 
\end{equation}
\begin{equation}  \label{6}
x_{t} = x^{(0)}_t +  \sigma x^{(1)}_t +\sigma ^{2} x^{(2)}_t + \cdots 
\end{equation}
%
%
The exogenous process $z_{t}$ can also easily be
represented in the form of expansion in $\sigma$
\begin{equation}  \label{7}
z_{t} =z_{t}^{(0)} +\sigma z_{t}^{(1)} .
\end{equation}
Indeed, plugging \eqref{7} into \eqref{2} gives 
\begin{equation*}
z_{t+1} =z_{t+1}^{(0)} +\sigma z_{t+1}^{(1)} =\Lambda (z_{t}^{(0)} +\sigma
z_{t}^{(1)} )+\sigma \varepsilon _{t+1} .
\end{equation*}
Collecting the terms of like powers of {$\sigma $} and equating them to
zero, we get 
\begin{align}
z_{t+1}^{(0)} &=\Lambda z_{t}^{(0)} ,\label{8}\\
z_{t+1}^{(1)} &=\Lambda z_{t}^{(1)} +\varepsilon _{t+1} .\label{9}
\end{align}
Since the expansion \eqref{7} must be valid for all $\sigma $ at
the initial time $t=0$, the initial conditions are
\begin{equation}  \label{9_}
 z_{0}^{(0)} =z_{0} \quad \mbox{and} \quad z_{0}^{(1)} =0.
\end{equation}

It is worth noting that in contrast to \cite{r19_}, where $y^{(0)}_t =\bar{y}$ and $x^{(0)}_t =\bar{x}$ are the steady state, here $x^{(0)}_t$ and $y^{(0)}_t$ are functions of $t$ and, as will be shown below, are a deterministic solution of the problem. In this way the pruning procedure  is always
local around the steady state, whereas we focus on the case where $x^{(0)}_t$ and $y^{(0)}_t$ are global. 

Then substituting \eqref{5}, \eqref{6} and \eqref{7} into \eqref{1},  we have
\begin{equation}\label{9__}
\begin{array}{l} 
E_{t}f(y^{(0)}_{t+1} + \sigma y^{(1)}_{t+1} + \sigma ^{2}y^{(2)}_{t+1} + \cdots, y^{(0)}_t +  \sigma y^{(1)}_t +\sigma ^{2} y^{(2)}_t + \cdots,\\
x^{(0)}_{t+1} +  \sigma x^{(1)}_{t+1} + \sigma ^{2}x^{(2)}_{t+1} + \cdots, x^{(0)}_t +  \sigma x^{(1)}_t +\sigma ^{2} x^{(2)}_t + \cdots,\\
z_{t+1}^{(0)} +\sigma z_{t+1}^{(1)},z_{t}^{(0)} +\sigma z_{t}^{(1)})  =  0.
\end{array}
\end{equation}

Expanding the left hand side of  \eqref{9__} for small $\sigma$, collecting the terms of
like powers of $\sigma $ and setting their coefficients  to zero, we obtain

 \textbf{Coefficient of {${\mathbf{\sigma }}^{0}$}}
\begin{equation}  \label{14}
f(y^{(0)}_{t+1}, y^{(0)}_{t},x_{t+1}^{(0)} ,x_{t}^{(0)} ,z_{t+1}^{(0)} ,z_{t}^{(0)})=0.
\end{equation}
The requirement that \eqref{6} and \eqref{7}
must hold for all arbitrary small $\sigma $ implies that the initial
conditions for \eqref{14} are
\begin{equation}  \label{15}
 z_{0}^{(0)} =z_{0}\quad \mbox{and} \quad x_{0}^{(0)} =x_{0}.   
\end{equation}
The terminal conditions for $y^{(0)}_{\infty}$ and $x^{(0)}_{\infty}$ are the deterministic steady states 
\begin{equation}  \label{15_}
 y^{(0)}_{\infty} =\bar{y}\quad \mbox{and} \quad x^{(0)}_{\infty} =\bar{x}.   
\end{equation}
%
The system of equations \eqref{8} and \eqref{14} is a deterministic model since it corresponds to the model \eqref{1} and \eqref{2}, where all shocks vanish (for this reason we omit the expectations operator in \eqref{14}). The deterministic model \eqref{8} and \eqref{14} with the initial and terminal conditions \eqref{15} and \eqref{15_}, respectively, can be solved globally by a number of effective algorithms, for example the extended path method ( \cite{r15}) or a Newton-like method (for example, \cite{r23}).
As this study is primarily concerned with stochastic models, in what follows we suppose that the deterministic model is already solved and its solution  is known.
%

 \textbf{Coefficient of {${\mathbf{\sigma }}^{1}$}}
\begin{equation}\label{first_order}
\begin{array}{l}
E_{t}\big\{ f_{1,t}\cdot y_{t+1}^{(1)} + f_{2,t}\cdot y_{t}^{(1)} +f_{3,t}x_{t+1}^{(1)}   
+f_{4,t}x_{t}^{(1)} \\
+ f_{5 ,t} z_{t+1}^{(1)}+ f_{6 ,t} z_{t}^{(1)}\big\}=0.
\end{array}%
\end{equation}
 The matrices
 \[f_{i,t} =f_{i}\left(y^{(0)}_{t+1},y^{(0)}_{t}, x_{t+1}^{(0)},x_{t}^{(0)} ,z_{t+1}^{(0)} ,z_{t}^{(0)} \right), \mbox{i=1,\ldots, 6},\]
are the Jacobian matrices  of the mapping $f$ with respect to the $i$th argument (that is $y_{t+1}$, $y_{t}$, $x_{t+1}$, $x_{t}$, $z_{t+1}$, and $z_{t}$, 
respectively), at the point $\left(y^{(0)}_{t+1},y^{(0)}_{t},x_{t+1}^{(0)} ,x_{t}^{(0)} ,z_{t+1}^{(0)},z_{t}^{(0)}\right).$
The requirement that \eqref{6} and \eqref{7}
must hold for all arbitrary small $\sigma $ implies that the initial
condition for \eqref{first_order} is
\begin{equation}
 z_{0}^{(1)} =0\quad \mbox{and} \quad  x_{0}^{(1)} =0. \label{19} 
\end{equation}
%

 \textbf{Coefficient of $\mathbf{\sigma }^{n}$, $n > 1$} 
\begin{equation}
\begin{array}{l}
E_{t} \Big\{ f_{1,t}\cdot y_{t+1}^{(n)} +f_{2,t}\cdot y_{t}^{(n)} +f_{3,t}\cdot x_{t+1}^{(n)} +f_{4,t}\cdot x_{t}^{(n)} +\eta _{t+1}^{(n)} \Big\} =0.\label{18}
\end{array}%
\end{equation}
The requirement that \eqref{6} 
must hold for all arbitrary small $\sigma $ implies that the initial
condition for \eqref{18} is
\begin{equation}
 x_{0}^{(n)} =0. \label{19_} 
\end{equation}

%
%
A nice feature of the set of systems of equations \eqref{18} is that the linear homogeneous part $f_{i,t}$ is
the same for all $n > 0$. The difference is only in the non-homogeneous terms $E_{t} \eta _{t+1}^{(n)}$ that are some mappings for which the set of arguments includes only quantities of  order less than $n$
\[ \left(y^{(0)}_{t+1},y^{(0)}_{t}, x_{t+1}^{(0)} ,x_{t}^{(0)}
 ,\ldots,y^{(n-1)}_{t+1},y^{(n-1)}_{t}, x_{t+1}^{(n-1)}, x_{t}^{(n-1)}
,z_{t+1}^{(0)} ,z_{t}^{(0)} ,z_{t+1}^{(1)} ,z_{t}^{(1)}\right).\]
Particularly, for $n=1,2$ we have
\begin{equation*}
E_{t} \eta _{t+1}^{(1)} =(f_{5 ,t} \Lambda+ f_{6 ,t}) z_{t}^{(1)},
\end{equation*}
and
\begin{equation}\label{quadr}
\begin{array}{l}
E_{t} \eta _{t+1}^{(2)} =
E_{t}\big\{\frac{1}{2}f_{11,t} \left(y_{t+1}^{(1)}\right)^2+  \frac{1}{2}f_{22,t} \left(y_{t}^{(1)}\right)^2+ \frac{1}{2}f_{33,t} \left(x_{t+1}^{(1)}\right)^2\\
+  \frac{1}{2}f_{44,t} \left(x_{t}^{(1)}\right)^2+\frac{1}{2}f_{55,t} \left(z_{t+1}^{(1)}\right)^2+  \frac{1}{2}f_{66,t} \left(z_{t}^{(1)}\right)^2
+f_{12,t} y_{t+1}^{(1)}y_{t}^{(1)}\\
 +f_{13,t} y_{t+1}^{(1)}x_{t+1}^{(1)}
+ f_{14,t} y_{t+1}^{(1)}x_{t}^{(1)}   +f_{15,t} y_{t+1}^{(1)}z_{t+1}^{(1)}+f_{16,t} y_{t+1}^{(1)}z_{t}^{(1)}\\ 
 +f_{23,t} y_{t}^{(1)}x_{t+1}^{(1)}+f_{24,t} y_{t}^{(1)}x_{t}^{(1)}    +f_{25,t} y_{t}^{(1)}z_{t+1}^{(1)}+f_{26,t} y_{t}^{(1)}z_{t}^{(1)}\\
+f_{34,t} x_{t+1}^{(1)}x_{t}^{(1)}+f_{35,t} x_{t+1}^{(1)}z_{t+1}^{(1)}
+f_{36,t} x_{t+1}^{(1)}z_{t}^{(1)}
+f_{45,t} x_{t}^{(1)}z_{t+1}^{(1)}\\ + f_{46,t} x_{t}^{(1)}z_{t}^{(1)}
+f_{56,t} z_{t+1}^{(1)}z_{t}^{(1)}   \big\},
\end{array}%
\end{equation}
respectively; where $f_{ij,t}$, $i=1,\ldots, 6$, $j=1,\ldots, 6$, denotes the mixed partial Frech\'{e}t derivative of $f_{t}$ of order two with
respect to $i$th and $j$th arguments at the
point
\begin {equation}\label{19__}
\left(y^{(0)}_{t+1},y^{(0)}_{t}, x_{t+1}^{(0)} ,x_{t}^{(0)},z_{t+1}^{(0)} ,z_{t}^{(0)}\right).
\end{equation}
In other words, $f_{ij,t}$ is a bilinear mapping (see, for
example,  \citet[p.~55]{r22}) depending on vector \eqref{19__} (and hence on $t$).\footnote{We do not make use tensor notation for brevity.}
The expectations $E_{t} \eta _{t+1}^{(n)}$ are  bounded if  all conditional mixed moments of $z_{t+1}^{(1)} $ are bounded up to
order $n$ and the vectors \eqref{19__}
are bounded for all $t \ge 0$.

Equation \eqref{18} with the initial conditions \eqref{19_} is a linear rational expectations model with time-varying
parameters and bounded the non-homogeneous terms $E_{t} \eta _{t+1}^{(n)}$. To solve the problem \mbox{\eqref{18}--\eqref{19}} is equivalent to finding
 a bounded solution $(x_{t}^{(n)} ,y_{t}^{(n)} )$  for $t>0$
under the assumption that the bounded solutions to 
the problems of all  orders less than $n$ are already known. Knowing how to solve these types of model and using the structure of mappings $E_{t} \eta _{t+1}^{(n)}$, we can find recursively solutions, $y_{t}^{(n)}, x_{t}^{(n)}
$, to \eqref{18} for every order $n$, starting with $n=1$.

\subsection{An Example of the Series Expansion: An Asset Pricing Model}\label{example}

In this section the method of expansion around a deterministic path applies to a simple  nonlinear asset
pricing model proposed by \cite{r16} and analyzed by \cite{r3,r14}. 
In this model the representative agent maximizes the lifetime utility function 
\begin{equation*}
\max \left(E_{0} \sum _{t=0}^{\infty }\beta ^{t} \frac{C_{t}^{\theta } }{%
\theta } \right)
\end{equation*}
subject to 
\[ p_{t} e_{t+1} +C_{t} =p_{t} e_{t} +d_{t} e_{t} ,\]
where $\beta>0$ is a subjective discount factor, $\theta<1$ and $\theta \ne 0$,  $C_{t}$ denotes consumption, $p_{t}$ is the price at date $t$ of
a unit of the asset
, $e_{t}$  represents  units of a single asset held at the beginning of period $t$,  and $ d_{t}$ is
dividends per asset in period $t$. The growth of rate of the dividends follows an  AR(1) process\footnote{By abuse of our previous  notation, we let $x_t$ stand for the exogenous process as in \cite{r16, r3, r14}. } 
\begin{equation}  \label{48}
\mathrm{x}_{\mathrm{t}} \mathrm{\; =\; (1-\; }\rho \mathrm{)\; }\bar{\mathrm{%
x}}\mathrm{\; }+\mathrm{\; }\rho x_{t-1} +\sigma \varepsilon _{t+1} ,
\end{equation}
where $x_{t} = ln(d_{t}/ d_{t-1})$, and $\varepsilon_{t+1 } \sim  NIID(0,1)$.
The first order condition and market clearing yields the equilibrium condition
\begin{equation}  \label{49}
y_{t} =\beta E_{t} \left[\exp (\theta x_{t+1} )\left(1+y_{t+1} \right)\right],
\end{equation}
where $y_{t} = p_{t}/d_{t}$ is the price-dividend ratio. This equation has an exact solution of the form \citep{r16} 
\begin{equation}  \label{50}
y_{t} =\sum _{i=1}^{\infty }\beta ^{i} \exp \left [a_{i} +b_{i} ( x_{t} -\bar{x})\right],
\end{equation}
where 
\begin{equation}  \label{51}
a_{i} =\theta \bar{x}i+\frac{1}{2} \left(\frac{\theta \sigma }{1-\rho }
\right)^{2} \left[i-\frac{2\rho (1-\rho ^{i} )}{1-\rho } +\frac{\rho ^{2}
(1-\rho ^{2i} )}{1-\rho ^{2} } \right]
\end{equation}
 and 
\begin{equation*}
b_{i} =\frac{\theta \rho (1-\rho ^{i} )}{1-\rho } .
\end{equation*}
 It follows from \eqref{49} that the deterministic steady state of
the economy is 
\begin{equation*}
\bar{y}=\frac{\beta \exp (\theta \bar{x})}{1-\beta \exp (\theta \bar{x})} .
\end{equation*}

 We now express a solution to the system \eqref{48}--\eqref{49} as an  expansion in powers of  the
parameter $\sigma $ up to a second-order approximation and decompose the original problem into a set of auxiliary problems. 
Specifically, assume that the solution can be represented in the form:
\begin{align}
y_{t} &=y^{(0)}_{t} +\sigma y^{(1)}_{t} +\sigma ^{2}
y^{(2)} _{t} \label{55}\\
x_{t} &=x_{t}^{(0)} +\sigma x_{t}^{(1)} .\label{56}
\end{align}
Substituting \eqref{56} into \eqref{48} and collecting
the terms containing $\sigma^{0}$  and $\sigma^{1}$, we obtain the representation \eqref{56} for $x_{t}$ 
\begin{align}
x_{t+1}^{(0)} &=\mathrm{(1-\; }\rho \mathrm{)\; }\bar{\mathrm{x}}\mathrm{\; }%
+\rho x_{t}^{(0)}\label{57}\\
x_{t+1}^{(1)} &=\rho x_{t}^{(1)} +\varepsilon _{t+1}.
\end{align}
Since the expansion \eqref{56} must be valid for all $\sigma $ at
the initial time $t=0$, the initial conditions are
\begin{equation}  \label{59}
 x_{0}^{(0)} =x_{0} \quad \mbox{and} \quad x_{0}^{(1)} =0.
\end{equation}


 Substituting now \eqref{55} and \eqref{56} into \eqref{49} yields
\begin{equation*}\
\begin{split}
&y^{(0)}_{t} +\sigma y^{(1)}_{t} +\sigma ^{2} y^{(2)}_{t}+\cdots \\ 
&=\beta E_{t} \big\{ \exp \left[\theta \left(x_{t+1}^{(0)} +\sigma x_{t+1}^{(1)}\right )\right]\big[1+y^{(0)}_{t}+\sigma y^{(1)}_{t}+\sigma ^{2} y^{(2)} +\cdots \big]\big\} 
\end{split}
\end{equation*}
 Expanding exponential for small $\sigma $ gives
 \begin{equation*}
\begin{split}
&y^{(0)}_{t} +\sigma y^{(1)}_{t} +\sigma ^{2} y^{(2)}_{t}+\cdots\\
&=\beta E_{t} \exp (\theta x_{t+1}^{(0)} )\left[ 1+\sigma \theta x_{t+1}^{(1)} +\frac{1}{2} \left(\sigma \theta x_{t+1}^{(1)}\right )^{2} +\cdots \right]\big[
1+y^{(0)}_{t+1}+\sigma y^{(1)}_{t+1}+\sigma ^{2} y^{(2)}_{t+1} +\cdots\big]%
\end{split}
\end{equation*}
Collecting the terms of like powers of $\sigma$ in the last equation, we have
%

\textbf{Coefficient of $\sigma^{0}$}
\begin{align}
y_{t}^{(0)}& =\beta \exp (\theta x_{t+1}^{(0)})(1+y_{t+1}^{(0)}),\label{60}\\
x_{t+1}^{(0)} &=\rho x_{t}^{(0)}. \label{61}
\end{align}

\textbf{Coefficient of $\sigma^{1}$}
\begin{equation}  \label{62}
\begin{split}
y_{t}^{(1)}=\exp (\theta x_{t+1}^{(0)} )\beta \theta \left(1+y_{t+1}^{(0)} \right)E_{t} 
x_{t+1}^{(1)} +\exp (\theta x_{t+1}^{(0)} )\beta E_{t} y_{t+1}^{(1)},
\end{split}
\end{equation}
\begin{equation}  \label{62:2}
x_{t+1}^{(1)} =\rho x_{t}^{(1)}+\varepsilon _{t+1}.
\end{equation}

\textbf{Coefficient of $\sigma^{2}$}
\begin{equation} \label{63}
\begin{array}{l}
y_{t}^{(2)} =
\frac{1}{2}\beta \exp (\theta x_{t+1}^{(0)} )\theta ^{2} \left(1+y_{t+1}^{(0)}\right)E_{t}\left(x_{t+1}^{(1)} \right)^{2} \\
+\theta \beta \exp (\theta x_{t+1}^{(0)} )E_{t} \left(x_{t+1}^{(1)} y_{t+1}^{(1)}\right)+\beta \exp (\theta x_{t+1}^{(0)} )E_{t} \left(y_{t+1}^{(2)}\right).
\end{array}%
\end{equation}

The system \eqref{60} and \eqref{61} is a deterministic model. Its solution can easily be obtained by, for example, forward induction
\begin{equation}  \label{64}
y_{t}^{(0)} =\sum _{i=1}^{\infty }\beta ^{i} \exp \left\{ \theta \left[\bar{x}i+\frac{%
\rho (1-\rho ^{i} )}{1-\rho } ( x_{t} -\bar{x})\right]\right\}.
\end{equation}

Under the assumption that $y_{t}^{(0)}$ and $x_{t}^{(0)} $ are already
known for $t \geqslant 0$, Equations \eqref{62} and %
\eqref{62:2} constitute  a linear  rational expectations model with time varying deterministic coefficients $\exp (\theta x_{t+1}^{(0)} )\beta$. The expectations of the term $\eta_{t+1}^{(1)}$ in \eqref{18} has the form $\exp (\theta x_{t+1}^{(0)} )\beta E_{t} \left[\theta
x_{t+1}^{(1)} (1+y_{t+1}^{(0)} )\right]$.
Equation \eqref{63} is also a linear  forward-looking equation with time varying deterministic coefficients $\exp (\theta x_{t+1}^{(0)} )\beta$, and the term 
\[     E_t \eta_{t+1}^{(2)} =   
\beta \exp (\theta x_{t+1}^{(0)} )\theta \left[\frac{1}{2}\theta \left(1+y_{t+1}^{(0)}\right)E_{t}\left(x_{t+1}^{(1)} \right)^{2} 
+E_{t} \left(\theta x_{t+1}^{(1)} y_{t+1}^{(1)}\right)\right]\] 
depending only on solutions of orders less than two, i.e. $x_{t+1}^{(0)}, y_{t+1}^{(0)}, x_{t+1}^{(1)}, y_{t+1}^{(1)}$. Therefore, both the system  \eqref{62} and 
\eqref{62:2}, and Equation \eqref{63} are linear  forward-looking models with time varying coefficients. Under the condition that we know how to solve these types of model, they can be solved recursively starting with solving \eqref{62} and 
\eqref{62:2}, then passing to \eqref{63}. 
In Section~\ref{sec6} we present a method for solving such types of model and  prove the convergence of the solutions implied by the method 
to the exact solution. In the next section we transform equation \eqref{18} in a more convenient form to deal with.

\section{Transformation of the Model}\label{transform} 

Define the deterministic steady state as vectors $(\bar{y}, \bar{x},0)$ such that 
\begin{equation}
f(\bar{y},\bar{y},\bar{x},\bar{x},0,0)=0.\label{22_}
\end{equation}
We can represent $f_{i,t}$ in \eqref{18} as $f_{i,t}=f_{i}+\hat{f}_{i,t}$, $ i = 1,\ldots,6$,  where 
\[{f_{i} =f_{i}(\bar{y},\bar{y},\bar{x},\bar{x},0,0)}\]
are the Jacobian matrices  of the mapping $f$ at the steady state with respect to $i$th argument, 
and 
\begin{equation}\label{22__}
\hat{f}_{i,t} =f_{i,t}(y^{(0)}_{t+1} ,y^{(0)}_{t},x_{t+1}^{(0)},x_{t}^{(0)}
,z_{t+1}^{(0)},z_{t}^{(0)} )-f_{i}(\bar{y},\bar{y},\bar{x},\bar{x},0,0).
\end{equation}
Note also that $\hat{f}_{i,t} \to 0$ as $t \to \infty$,  because a deterministic
solution must tend to the deterministic steady state as $t$ tends to infinity. 
Consequently, $f_{i,t} $ can be thought of as a perturbation of $f_{i}$. 
As Equations \eqref{18} have the same form for all $n>0$, to shorten notation, further on we omit
the superscript $(n)$ when no confusion can arise. 
Therefore Equations \eqref{18} can be written in the vector form  
\begin{equation}\label{22}
\Phi _{t} E_{t} \left[%
\begin{array}{c}
{x_{t+1} } \\ 
{y_{t+1} }%
\end{array}%
\right]=\Lambda _{t} \left[
\begin{array}{c}
{x_{t} } \\ 
{y_{t} }%
\end{array}%
\right]+E_{t} \eta _{t+1}, 
\end{equation}
where $\Phi _{t} =\left[f_{3} +\hat{f}_{3,t} ,f_{1} +\hat{f}%
_{1,t} \right]$ and $\Lambda _{t} =\left[f_{4} +\hat{f}_{4,t} ,f_{2} +\hat{f%
}_{2,t} \right]$. We assume that the matrices $\Phi _{t}$ are
invertible for all $t\ge 0$. For instance, this assumption always holds in some neighborhood of the steady state if  the Jacobian $\left[f_{3}
,f_{1}\right]^{-1}$ at the steady state  is invertible.\footnote{This assumption is made for ease of
exposition. If  $[f_{3} ,f_{1} ]$ is a singular matrix, then further on we
must use a generalized Schur decomposition for which derivations remain
valid, but become more complicated.} 

Pre-multiplying \eqref{22} by {$\Phi _{t}^{-1} $}, we get 
\begin{equation}  \label{24}
E_{t} \left[%
\begin{array}{c}
{x_{t+1} } \\ 
{y_{t+1} }%
\end{array}%
\right]=L\left[%
\begin{array}{c}
{x_{t} } \\ 
{y_{t} }%
\end{array}%
\right]+M_{t}\left[
\begin{array}{c}
{x_{t} } \\ 
{y_{t} }%
\end{array}%
\right]+\Phi _{t}^{-1} E_{t} \eta _{t+1},
\end{equation}
where $L=\left[f_{3} ,f_{1} \right]^{-1}\left [f_{4} ,f_{2} \right]$ and 
\begin{equation*}
M_{t}=\left[f_{3} +\hat{f}_{3,t} ,f_{1} +\hat{f}_{1,t} \right]^{-1} \left[f_{4} +\hat{f%
}_{4,t} ,f_{2} +\hat{f}_{2,t}\right ]-\left[f_{3} ,f_{1} \right]^{-1} \left[f_{4} ,f_{2}\right].
\end{equation*}
Particularly, for $n=1$ we have
\begin{equation}  \label{24_}
E_{t} \left[%
\begin{array}{c}
{x_{t+1}^{(1)} } \\ 
{y_{t+1}^{(1)} }%
\end{array}%
\right]=L\left[%
\begin{array}{c}
{x_{t}^{(1)} } \\ 
{y_{t}^{(1)} }%
\end{array}%
\right]+M_{t}\left[
\begin{array}{c}
{x_{t}^{(1)} } \\ 
{y_{t}^{(1)} }%
\end{array}%
\right]+\Phi _{t}^{-1} (f_{5 ,t}\Lambda + f_{6 ,t}) z_{t}^{(1)}.
\end{equation}
Notice that $ lim_{t \to \infty}M_{t}=0$.  
As in the case of  rational expectations models with constant parameters it is
convenient to transform \eqref{24} 
using the spectral property of  $L$.  
Namely, the matrix $L$ is transformed into a block-diagonal one
%
\begin{equation}  \label{25}
 L=ZPZ^{-1} , 
\end{equation}
 where
\begin{equation} \label{25_}
P=\left[%
\begin{array}{cc}
{A} & {0} \\ 
{0} & {B}%
\end{array}%
\right],
\end{equation}
where $A$ and $B$ are 
matrices with eigenvalues larger and smaller than one (in modulus), respectively;
and $Z$ is an invertible matrix\footnote{A simple Schur triangular factorization is also possible to be employed here,
but at the cost of more complicated derivations. The block-diagonal structure of the matrix $P$ simplifies algebra}. This can be done, for example, by  initially transforming $L$ in a simple Schur form $L=Z_1 L_1 Z_1^{-1}$, where $Z_1$ is a unitary matrix, $L_1$ is an upper triangular Schur form with
the eigenvalues along the diagonal. We then transform the matrix $L_1$ in the block-diagonal Schur factorization $L_1=Z_2 P Z_2^{-1}$, where $Z_2$ is an invertible matrix and $P$ is block-diagonal and each diagonal block is a quasi upper-triangular Schur matrix\footnote{The function bdschur of Matlab Control System Toolbox performs this factorization.}. Hence the matrix $Z$ in \eqref{25} has the form $Z=Z_1Z_2$. 
We
also impose the conventional Blanchard-Kan condition (\cite{r25}) on the dimension of the unstable subspace
, i.e., $dim (B) =n_y$.

 After introducing the auxiliary variables 
\begin{equation}\label{26_0}
[s_{t} ,u_{t} ]^{\prime }=Z^{-1} [x_{t} ,y_{t} ]^{\prime }
\end{equation}
and pre-multiplying %
\eqref{24} by $Z^{-1} $, we have 
\begin{align}
E_{t} s_{t+1} =As_{t} +Q_{11,t} s_{t} +Q_{12,t} u_{t} +\Psi
_{1t} E_{t} \eta _{t+1},\label{26}\\
E_{t} u_{t+1} =Bu_{t} +Q_{21,t} s_{t} +Q_{22,t} u_{t} +\Psi
_{2t} E_{t} \eta _{t+1},\label{26_}
\end{align}
where $[\Psi_{1,t}, \Psi_{2,t}] = Z \Phi_{t}^{-1} $ and 
\begin{equation}
 \left[%
\begin{array}{cc}
{Q_{11,t} } & {Q_{12,t} } \\ 
{Q_{21,t} } & {Q_{22,t} }%
\end{array}%
\right]=ZM_{t} Z^{-1}.\label{26_2}
\end{equation}
Particularly, for $n=1$, we have
\begin{align*}
E_{t} s_{t+1}^{(1)} =As_{t}^{(1)} +Q_{11,t} s_{t}^{(1)} +Q_{12,t} u_{t}^{(1)} +{\Pi}_{1,t}z_{t}^{(1)},\\
E_{t} u_{t+1}^{(1)} =Bu_{t}^{(1)} +Q_{21,t} s_{t}^{(1)} +Q_{22,t} u_{t}^{(1)} +{\Pi}_{2,t}z_{t}^{(1)},
\end{align*}
where 
\begin{equation*}
 \left[%
\begin{array}{cc}
{\Pi}_{1,t}  \\ 
{\Pi}_{2,t} %
\end{array}%
\right]=\Phi _{t}^{-1} (f_{5 ,t}\Lambda + f_{6 ,t}) .
\end{equation*}

System \eqref{26}-\eqref{26_} is a linear rational
expectations model with time-varying parameters, therefore to solve the system we cannot apply 
the approaches used in the case of 
models with constant parameters (\cite{r25,r30, r32, r5_,r31},  etc.).
In Subsection~\ref{sec6.1} we 
develop a method for solving this type of models. 

 \section{Solving the Rational Expectations Model with Time-Varying Parameters} \label{sec6}
\subsection{Notation}\label{not}
 This subsection introduces some notation that will be necessary further on. 
By $\lvert \cdot \rvert$ denote the Euclidean norm in ${\mathbb R}^{n}$. The induced norm for a real matrix $D$ is
 defined by
\begin{equation*}
\lVert D\rVert =\sup_{\lvert s\rvert = 1}
 \lvert \ Ds \rvert.
\end{equation*}
The matrix $Z$ in (\ref{25})   can be chosen in such a way that 
\begin{equation}\label{27__}
\lVert A\rVert < \alpha + \gamma < 1  \text{ and }  \lVert B^{-1}\rVert < \beta + \gamma < 1,
\end{equation}
where $\alpha$ and $\beta$ are the largest eigenvalues (in modulus) of the matrices $A$ and $B^{-1}$, respectively, and $\gamma$  is arbitrarily small. 
This follows from the same
arguments as in  \citet[\S IV 9]{r6}, where it is done for the Jordan matrix
decomposition. 
Note also that ${\lVert B\rVert}^{-1} < 1$ for sufficiently small $\gamma$. Let 
\begin{equation}
B_{t} =B+Q_{22,t},  \quad  A_{t} =A+Q_{11,t}. \label{5.2}
\end{equation}
By definition, put
\begin{align}
&a=\mathop{\sup }\limits_{t=0,1,\dots} \left\| A_{t} \right\|,\quad b=\mathop{\sup }%
\limits_{t=0,1,\ldots} \left\| B_{t}^{-1} \right\|  , \label{27_0}\\
&c=\mathop{\sup }\limits_{t=0,1,\ldots} \left\| Q_{12,t} \right\|, \quad d=%
\mathop{\sup }\limits_{t=0,1,\ldots} \left\| Q_{21,t} \right\| .\label{27}
\end{align}
In the sequel, we assume that all the matrices $B_{t}$, $t=0,1,\ldots,$ are invertible. Note that the
numbers $a$, $b$, $c$ and $d$ depend on the initial conditions $(x_{0}^{(0)} ,z_{0}^{(0)})$. From the definitions of 
$A_{t}$, $A$, $B_{t}$, $B$, $Q_{12,t}$ and $Q_{21,t}$ and the condition $lim_{t \to \infty}(x_{t}^{(0)} ,z_{t}^{(0)})= (\bar{x},0)$, it follows that
\begin{equation}  \label{28}
\mathop{\lim }\limits_{t\to \infty}c(x_{t}^{(0)} ,z_{t}^{(0)})=0,\quad \mathop{\lim }\limits_{t\to \infty} d(x_{t}^{(0)} ,z_{t}^{(0)} )=0,
\end{equation}
\begin{equation*}
\mathop{\lim }\limits_{t\to \infty}
a(x_{t}^{(0)} ,z_{t}^{(0)} )=\left\| A\right\| <1,\quad \mathop{\lim }\limits_{t\to \infty
} b(x_{t}^{(0)}
,z_{t}^{(0)} )=\left\| B^{-1} \right\| <1.
\end{equation*}
This means that $c$ and $d$ can be arbitrary small and 
\begin{equation} 
 a <1 \quad \mbox{and}\quad b < 1 \label{29}
\end{equation}
by choosing $(x_{0}^{(0)} ,z_{0}^{(0)} )$ close enough to
the steady state.

\subsection{Solving the transformed system \eqref{26}--\eqref{26_} }\label{sec6.1}
Taking into account notation \eqref{5.2}, we can rewrite \eqref{26}--\eqref{26_} in the form
\begin{align}
E_{t} s_{t+1} =A_{t} s_{t} +Q_{12,t} u_{t} +\Psi _{1,t} E_{t}
\eta _{t+1} ,\label{30}\\
E_{t} u_{t+1} =B_{t} u_{t} +Q_{21,t} s_{t} +\Psi _{2,t} E_{t}
\eta _{t+1} \label{31}.
\end{align}
In this subsection we construct a bounded solution to \eqref{30}--\eqref{31} for $t\ge 0$ with an arbitrary initial condition $s_{0} \in {\mathbb{R}}^{n_{x}}$ and find under which conditions this solution exists.  For
this purpose,  we first start
with solving a finite-horizon problem with a fixed terminal condition using backward recursion. Then, we prove
the convergence of the obtained finite-horizon solutions to a bounded infinite-horizon one as the terminal time $T$ tends to infinity. 

Fix a horizon $T > 0$. 
At the time $T$ using the  invertibility of $B_{T}$ 
and solving  Equation \eqref{31} backward, we can obtain $u_{T }$ as a linear function of  $s_{T}$, the terminal condition $E_{T} u_{T+1} $ and the ``exogenous'' term $\Psi
_{2,T} E_{T} \eta _{T+1} $ 
\begin{equation*}
u_{T} =-B_{T}^{-1} Q_{21,T} s_{T} -B_{T}^{-1} \Psi _{2,t} E_{T}
\eta _{T+1} +B_{T}^{-1} E_{T} u_{T+1} .
\end{equation*}
Proceeding further with  backward recursion, we shall obtain finite-horizon
solutions for each $t = 0,1,2,\dots, T.$ For doing this we need to
define the following recurrent sequence of matrices:  
\begin{equation}  \label{32}
K_{T,T-i-1} =L_{T,T-i}^{-1} \left(Q_{21,T-i} +K_{T,T-i} A_{T-i}\right ), \quad i = 0,1,\dots ,T, 
\end{equation}
where 
\begin{equation}  \label{33}
L_{T,T-i} =B_{T-i}+ K_{T,T-i} Q_{12,T-i},
\end{equation}
with the terminal condition $K_{T,T+1} = 0$. In \eqref{32} and \eqref{33} the first subscript $T$ defines the time horizon, while the second subscript defines all times between $0$ and 
$T+1$. Let  $u_{T,T-i}$, $i=0,1,\ldots, T,$ denote the $(T-i)$-time solution obtained by backward recursion  that starts at the time $T$. The matrices \eqref{32} and \eqref{33} are needed for constructing approximate solutions by backward recursion. 

\begin{proposition}\label{prop1}
Suppose that the sequence of matrices %
\eqref{32} and  \eqref{33} exists; then the solution to \eqref{30}--\eqref{31} has the following representation:
\begin{equation}  \label{35}
 u_{T,T-i} =-K_{T,T-i} s_{T-i} +g_{T,i} + \left(\prod _{k=1}^{i+1}L_{T,T-i+k}^{-1}\right)
E_{T-i} \left(u_{T+1} \right),
\end{equation}
 where $i = 0, 1, \dots , T;$ and 
\begin{equation}  \label{36}
g_{T,i} =-\sum _{j=1}^{i+1}
\prod _{k=1}^{j}L_{T,T-i+k}^{-1}
(\Psi _{2,T-i+j}
+K_{T,T-i+j} \Psi _{1,T-i+j} )E_{T-i}
\eta _{T-i+j}.
\end{equation}
\end{proposition}
For the proof see 
Appendix~\ref{A}. 

The sequence of matrices \eqref{32} exists if all matrices 
$L_{T,T-i}$, $i=0,1,\dots ,T, $ are invertible. For this
we need, in addition, some boundedness condition on
the matrices $B_{T-i}^{-1}K_{T,T-i+1}Q_{12,T-i}$. From \eqref{28} the  matrices $B_{T-i}^{-1}$ and $Q_{12,T-i}$ are bounded, hence this condition boils down to the boundedness of matrices $K_{T,T-i+1}$.
\begin{proposition}\label{prop2}
If for $a$, $b$, $c$ and $d$ from \eqref{27_0}--\eqref{27} the inequality
\begin{equation}  \label{38}
cd<\frac{1}{4}\left(\frac{1}{b}-a\right)^{2}=\left(\frac{1-ab}{2b}\right)^{2}
\end{equation}
holds, then 
\begin{equation}\label{38_}
\left\| B_{T-i}^{-1} \right\| \cdot \left\| K_{T,T-i+1} \right\| \cdot\left\| Q_{12,T-i} \right\| <1, \quad i=0,1,2,\ldots T.
\end{equation}
\end{proposition}
 For the proof see Appendix~\ref{A}. 
\begin{proposition}\label{prop3}
If the inequality \eqref{38_} holds, then the matrices $L_{T,T-i}$, $i=0,1,2.\ldots,T$, are invertible.
\end{proposition}
\begin{proof}
From \eqref{33} and the invertibility of $B_{T-i}$ it follows that
\begin{equation}  
L_{T,T-i} =B_{T-i}\left(I+ B_{T-i}^{-1}K_{T,T-i} Q_{12,T-i}\right).
\end{equation}
The matrices $L_{T,T-i}$ are  invertible if and only if the matrices $\left(I+ B_{T-i}^{-1}K_{T,T-i} Q_{12,T-i}\right)$ are invertible.
From the norm property and \eqref{38_} we have
\begin{equation*}
\left\| B_{T-i}^{-1} K_{T,T-i+1} Q_{12,T-i} \right\| \le \left\|
B_{T-i}^{-1} \right\| \cdot \left\| K_{T,T-i+1} \right\| \cdot \left\|
Q_{12,T-i} \right\| <1. 
\end{equation*}
The invertibility of $\left(I+ B_{T-i}^{-1}K_{T,T-i} Q_{12,T-i}\right)$ now follows   from \citet[Lemma 2.3.3]{r26}. 
\end{proof}
 For $i=T$ from \eqref{35} we have 
\begin{equation}  \label{39}
 u_{T,0} =-K_{T,0} s_{0} +g_{T,T} + \left(\prod _{k=1}^{T+1}L_{T,k}^{-1}\right)
E_{0} \left(u_{T+1} \right).
\end{equation}
This is a finite-horizon solution to the rational expectations model with time-varying coefficients \eqref{30}--\eqref{31} and with a given initial condition $s_0$.
What is left is to show that the solution $u_{T,0}$ of the form \eqref{39} converges to some limit as $T \to \infty$.
\begin{proposition}\label{prop4}
If  inequality \eqref{38} holds, then the limit  
\[\mathop{\lim }%
\limits_{T\to \infty } K_{T,j} =K_{\infty ,j}  \quad \mbox{for} \quad j=0,1,2,\ldots\]
exists in the matrix space defined in Subsection~\ref{not}. 
\end{proposition}
 For the proof see Appendix~\ref{A}. 
\begin{proposition}\label{prop5}
If  inequality \eqref{38_} holds, then 
 \begin{equation}  \label{40}
\lim_{T \to \infty} \prod _{k=1}^{T+1}L_{T,k}^{-1} = 0
\end{equation}
 and
  \begin{equation}  \label{41}
\lim_{T \to \infty} g_{T,T}= g_{\infty},
\end{equation}
where $g_{\infty}$ is some vector in $\mathbb{R}^{n_{y}}$.
 \end{proposition}
 \begin{proof}
 From \eqref{33} and \ref{prop4} it follows that   
 \[\lim_{T \to \infty} L_{T,k}=B_{k}+ K_{\infty,k} Q_{12,k} = L_{\infty,k}. \]
Then the limit in \eqref{40} can be represented as
 \begin{equation}  \label{42}
\lim_{T \to \infty} \prod _{k=1}^{T+1}L_{T,k}^{-1} =\lim_{T \to \infty} \prod _{k=1}^{T+1} L_{\infty,k} ^{-1}.  
\end{equation}
Since $K_{\infty,k}$ is bounded (it follows from formula \eqref{76} in \ref{A}) and
\[\lim_{k \to \infty} Q_{12,k}=0,\quad   \text{and}\quad   \lim_{k \to \infty}B_{k}^{-1}=B^{-1},\]  
 we have $lim_{k \to \infty}L_{\infty,k}^{-1} =B^{-1}$.
Therefore, if $\delta >0$ is arbitrary small, there is an $N=N_{\delta}\in {\mathbb{N}}$ such that  
 \begin{equation}  \label{42_}
 \lVert L_{\infty,k}^{-1}\rVert \le  \beta + \delta =\rho < 1, 
\end{equation}
for $k >  N$, where $\beta$ is the largest eigenvalue (in modulus) of the matrix $B^{-1}$. 
From this, the norm property and \eqref{42} we obtain
\[\lim_{T \to \infty} \left\lVert \prod _{k=1}^{T+1}L_{T,k}^{-1}\right\rVert \le \lim_{T \to \infty}  \prod _{k=1}^{T+1}\left\lVert L_{\infty,k}^{-1}\right\rVert \le\lim_{T \to \infty}  C_1 \rho^{T-K}=0,  \]
where $C_1$ is some constant. 

By \eqref{42_} the products in \eqref{36}    decay exponentially with the factor $\rho$ as $j \to \infty$.  From this and the  boundedness of the terms $ K_{T,k}$, $\Psi _{2,k}$, $\Psi _{1,k}$ and $E_{0}\eta _{k}$,  $T \in \mathbb{N}$ and $k=1,2,\ldots,T+1$, it follows 
 that the series
\begin{equation*}  
g_{T,T} =-\sum _{j=1}^{T+1}\prod _{k=1}^{j}L_{T,k}^{-1}
(\Psi _{2,j}
+K_{T,j} \Psi _{1,j} )E_{0}
\eta _{j}.
\end{equation*}
converges to some $g_{\infty}$ as $T \to \infty$. 
  \end{proof}
%

From  Proposition~\ref{prop4} and Proposition~\ref{prop5} it may be concluded that as $T$ tends to infinity  Equation \eqref{39} takes the form:
\begin{equation}  \label{44}
u_{0} =-K_{\infty ,0} s_{0} +g_{\infty }.
\end{equation}
Formula \eqref{44} provides a unique bounded solution to the transformed rational expectation model with time-varying parameters
\eqref{30}--\eqref{31}, and may be treated as a policy function for this type of problems. 
\begin{remark}\label{rem}
Particularly, for $n=1$ we have 
\begin{equation}  \label{g}
\begin{split}
g_{T,T}^{(1)} =-\sum _{j=1}^{T+1}\prod _{k=1}^{j}L_{T,k}^{-1}
(\Psi _{2,j}
+K_{T,j} \Psi _{1,j} )E_{0}z^{(1)} _{j}\\
= -\sum _{j=1}^{T+1}\prod _{k=1}^{j}L_{T,k}^{-1}
(\Psi _{2,j}
+K_{T,j} \Psi _{1,j} )\Lambda^{j+1} z^{(1)} _{0}.
\end{split}
\end{equation}
Taking into account \eqref{9_}, we get $g_{T,T}^{(1)} = 0$.

\end{remark}
\begin{remark}\label{rem0}
The details of derivations for the solution of time-varying rational expectations model corresponding to  the  first order approximation of the system~\eqref{18} and \eqref{19_} are carried out in Appendix~\ref{FOS2}, where we also derive  the moving-average representation for $x^{(1)}_{t }$ and  $y^{(1)}_{t }$. Having this representation  it is not hard to compute all quadratic terms in \eqref{quadr}.
\end{remark}

\begin{remark}\label{rem1}
If $c=0$ or $d=0$ (or both) in the inequality \eqref{38}, i.e. one of the variables $s_t$ or $u_t$ (or both) is exogenous to the other, then  \eqref{38} is always valid under the conditions \eqref{29}.
\end{remark}
\begin{remark}\label{rem2}
The inequality \eqref{38} is a sufficient condition for the existence of the solution in the form \eqref{44}, and can be weakened. For the representation  \eqref{44} we need only the invertibility of matrices $L_{T,T-i}$ defined in \eqref{33}.
\end{remark}

 \subsection{Initial conditions.}

It remains to find the initial condition for a stable solution to the system \eqref{30}--\eqref{31} corresponding the initial condition \eqref{19_}.
Recall that we deal with the $n$-order problem \eqref{18}--\eqref{19_}, and   we now put  the superscript $(n)$ back in notation. 
%
%
From \eqref{26_0} and
\eqref{44} we have
 \[\left[%
\begin{array}{c}
{s_{0}^{(n)} } \\ 
{-K_{\infty ,0}^{(n)} s_{0}^{(n)} +g_{\infty }^{(n)} }%
\end{array}%
\right]=Z^{-1} \left[%
\begin{array}{c}
{0} \\ 
{y_{0}^{(n)} }%
\end{array}%
\right],\]
where  $Z^{-1}$ is a matrix that is involved in the block-diagonal 
factorization \eqref{25} and has
the following block-decomposition: 
\[Z^{-1} =\left[%
\begin{array}{cc}
{Z^{11} } & {Z^{12} } \\ 
{Z^{21} } & {Z^{22} }
\end{array}%
\right].\]
 Hence 
\begin{align}
s_{0}^{(n)} =Z^{12} y_{0}^{(n)},\label{45}\\
-K_{\infty ,0}^{(n)} s_{0}^{(n)} +g_{\infty }^{(n)} =Z^{22} y_{0}^{(n)}.\label{46}
\end{align}

Substituting \eqref{45} into \eqref{46}
, we get 
\begin{equation}  \label{47}
y_{0}^{(n)} =(Z^{22} +K_{\infty ,0}^{(n)} Z^{12} )^{-1} g_{\infty }^{(n)} .
\end{equation}
The vector $(y_{0}^{(n)},0)$ is the initial condition corresponding to a bounded solution to \eqref{18} for $t>0$, hence formula \eqref{47} determines the solution to the original rational expectations model with time-varying parameters. In other words, $y_{0}^{(n)}$ is a policy function for the rational expectations model with time-varying parameters at the point  $x_{0}^{(n)} =0$. Particularly, 
 for $n=1$ from \eqref{47} and taking into account $g_{\infty}^{(1)} = 0$ we have $y_{0}^{(1)}=0$. The condition of the  invertibility of matrix  $Z^{22} +K_{\infty ,0}^{(n)} Z^{12}$  corresponds to Proposition 1 of  \cite{r25}. 

 \subsection{Expected dynamics. Restoring the original variables $x_{t}^{(n)}$ and $y_{t}^{(n)}$.}

To compute the expected dynamics (impulse response function) it is more convenient to work with auxiliary variables $u_t^{(n)}$ and $s_t^{(n)}$, then to restore the original variables $x_{t}^{(n)}$ and $y_{t}^{(n)}.$
Substituting \eqref{44} for $u_t$ in \eqref{30} and taking expectations at $t=0$ gives 
\begin{equation}\label{47:1}
E_{0}s_{t+1}^{(n)} = (A_t - Q_{12,t}K_{\infty ,t}^{(n)} )E_{0}s_{t}^{(n)} +Q_{12,t}g_{\infty}^{(n)}+\Psi_{1,t}E_{0}\eta_{t+1}^{(n)}.
\end{equation}
 From \eqref{45} we can compute the initial condition $s_{0}^{(n)}$ for \eqref{47:1}. Knowing the initial value $s_0^{(n)}$ allows us to obtain the whole trajectory of the solution to \eqref{47:1}, i.e. $E_0s_t, t \in \mathbb{N}$. The expected dynamics of $E_0u_t$ can easily be obtained from \eqref{44} 
\begin{equation}\label{47:2}
E_{0}u_{t}=-K_{\infty ,t}^{(n)} E_{0}s_{t}+g_{\infty}.
\end{equation}
Then the expected dynamics of the original variables is restored by 
\begin{align}
E_{0}x^{(n)}_{t } = Z_{11}E_{0}s^{(n)}_{t} + Z_{12}E_{0}u^{(n)}_{t}, \label{47:3}\\ 
E_{0}y^{(n)}_{t } = Z_{21}E_{0}s^{(n)}_{t} + Z_{22}E_{0}u^{(n)}_{t},\label{47:4}
\end{align}
 where $Z_{ij}$, $i=1,2$, $j=1,2,$ are blocks of the
block-decomposition of the matrix $Z$.
From \eqref{47:1} and \eqref{47:2} it follows that the process $(s_t, u_t)$ is stable, as $A_t \to A$, but $A$ is a stable matrix, $Q_{12,t} \to 0$, and $K_{\infty ,t}^{(n)}$ are bounded matrices for $t \geq 0$.From this and \eqref{47:3} and  \eqref{47:4} it may be concluded that the process $(y^{(n)}_{t } , x^{(n)}_{t } )$ is also stable.
To sum up, under the assumption that the solutions of lower order than $n$ are already computed in the same manner as for the $n$th order, we find the stable solution to the original model \eqref{1} in the form
\begin{align*}
E_{0}y_{t} =\sum _{i=0}^{n}\sigma ^{i} E_{0}y^{(i)}_{t}, \\
E_{0}x_{t} =\sum _{i=0}^{n}\sigma ^{i} E_{0}x^{(i)}_{t}. 
\end{align*}

 \section{Approximate solution: an Asset Pricing Model}\label{example}

To illustrate how the presented method works we apply it to the nonlinear asset
pricing model considered above. The simplicity of the model allows us to derive all approximations  in the analytical form.
%
We begin with the first order approximation determined by Equations \eqref{62} and%
\eqref{62:2} under the assumption that the deterministic solution $y_{t}^{(0)}$ and $x_{t}^{(0)} $ are
known for $t \geqslant 0$ and $y_{t}^{(0)}$  satisfies \eqref{64}. Rewriting \eqref{62} for $t=T$ and taking into account that $E_{T} 
x_{T+1}^{(1)}= \rho x_{T}^{(1)}$ gives

\begin{equation} 
\begin{split}
y_{T}^{(1)}=\exp (\theta x_{T+1}^{(0)} )\beta \theta \left(1+y_{T+1}^{(0)} \right)\rho x_{T}^{(1)} +\exp (\theta x_{T+1}^{(0)} )\beta E_{T} y_{T+1}^{(1)}, \label{T}
\end{split}
\end{equation}
Similarly to \eqref{T} for $t=T-1$ we have
\begin{equation} 
\begin{split}
y_{T-1}^{(1)}=\exp (\theta x_{T}^{(0)} )\beta \theta \left(1+y_{T}^{(0)} \right)\rho x_{T-1}^{(1)} +\exp (\theta x_{T}^{(0)} )\beta E_{T-1} y_{T}^{(1)}, \label{T-1}
\end{split}
\end{equation}
Substituting in the last equation \eqref{T} for $y_{T}^{(1)}$ and taking into account that $E_{T-1} 
x_{T}^{(1)}= \rho x_{T-1}^{(1)}$ , we obtain
\begin{equation*} 
\begin{split}
y_{T-1}^{(1)}=\big[\theta\beta \rho \exp (\theta x_{T}^{(0)} )\left(1+y_{T}^{(0)} \right)  +\theta(\beta\rho)^2 \exp (\theta (x_{T}^{(0)} +x_{T+1}^{(0)})) \left(1+y_{T+1}^{(0)} \right)\big]x_{T-1}^{(1)}\\
+\beta^2 \exp (\theta (x_{T}^{(0)} +x_{T+1}^{(0)}))E_{T-1} y_{T+1}^{(1)}). \label{T-1_}
\end{split}
\end{equation*}
Continuing further in the same way, for $t=T-k+1$ we have 
\begin{equation}   \label{T_i}
\begin{split}
y_t^{(1)}=\underbrace{ \theta \left[\sum _{i=1}^{k }(\beta \rho)^{i}\left( 1+  y_{t+i}^{(0)} \right)exp\left( \theta \sum _{j=1}^{i}x_{t+j}^{(0)} \right)\right]}_{= - K_{T,t} }x_{t}^{(1)}+ \\
\beta^{k} exp\left( \theta \sum _{j=1}^{i}x_{t+j}^{(0)} \right) E_{t} y_{T+1}^{(1)}.
\end{split}
\end{equation}
If  the  moment $T$ tends to $\infty$, then the following solution for $y_t^{(1)}$ is valid:
\begin{equation}   \label{65_}
y_t^{(1)}=\underbrace{ \theta \left[\sum _{i=1}^{\infty }(\beta \rho)^{i}\left( 1+  y_{t+i}^{(0)} \right)exp\left( \theta \sum _{j=1}^{i}x_{t+j}^{(0)} \right)\right]}_{= - K_{\infty,t} }
x_{t}^{(1)} = - K_{\infty,t} x_{t}^{(1)}.
\end{equation}
%
%
Note that $x_{0}^{(1)}=0$, hence  $y_{0}^{(1)}=0$.

%
%
%
%
%
We now turn to the second order approximation. 
Equation \eqref{63} is also a linear  forward-looking equation with time varying deterministic coefficients and can solved by the backward induction. Indeed, rewriting \eqref{63} for $t=T$ yields
\begin{equation} 
\begin{array}{l}
y_{T}^{(2)} =
\frac{1}{2}\beta \exp (\theta x_{T+1}^{(0)} )\theta ^{2} \left(1+y_{T+1}^{(0)}\right)E_{T}\left(x_{T+1}^{(1)} \right)^{2} \\
+\theta \beta \exp (\theta x_{T+1}^{(0)} )E_{t} \left(x_{T+1}^{(1)} y_{T+1}^{(1)}\right)+\beta \exp (\theta x_{T+1}^{(0)} )E_{T} \left(y_{T+1}^{(2)}\right).
\end{array}%
\end{equation}
Substituting \eqref{65_} for $y^{(1)}_{T+1}$ in \eqref{63} and collecting the terms with $E_{T}\left(x_{T+1}^{(1)} \right)^{2}$ yields
\begin{equation} \label{63__}
\begin{array}{l}
y_{T}^{(2)} =
\frac{1}{2}\beta \exp (\theta x_{T+1}^{(0)} )\theta \left [   \theta \left(1+y_{T+1}^{(0)}\right) - 2K_{\infty,T+1}\right]
E_{T}\left(x_{T+1}^{(1)} \right)^{2} \\
+\beta\exp (\theta x_{T+1}^{(0)} )E_{T} \left(y_{T+1}^{(2)}\right).
%
%
\end{array}%
\end{equation}
Substituting $T-1$ for $T$ in \eqref{63__} gives
\begin{equation} \label{632}
\begin{array}{l}
y_{T-1}^{(2)} =
\frac{1}{2}\beta \exp (\theta x_{T}^{(0)} )\theta \left [   \theta \left(1+y_{T}^{(0)}\right) - 2K_{\infty,T}\right]
E_{T}\left(x_{T}^{(1)} \right)^{2} \\
+\beta\exp (\theta x_{T}^{(0)} )E_{T-1} \left(y_{T}^{(2)}\right).
\end{array}%
\end{equation}
Inserting $y_{T}^{(2)}$ from \eqref{63__} into \eqref{632}, we have
\begin{equation} 
\begin{array}{l}
y_{T-1}^{(2)} =
\frac{1}{2}\beta \exp (\theta x_{T}^{(0)} )\theta \left [   \theta \left(1+y_{T}^{(0)}\right) - 2K_{\infty,T}\right]
E_{T}\left(x_{T}^{(1)} \right)^{2} \\
+\frac{1}{2}\beta^{2}\exp \left(\theta \left( x_{T}^{(0)}+x_{T+1}^{(0)}\right)\right )  \theta \left [   \theta \left(1+y_{T+1}^{(0)}\right) - 2K_{\infty,T+1}\right]
E_{T-1}\left(x_{T+1}^{(1)} \right)^{2} \\
+\beta^{2}\exp\left (\theta \left( x_{T}^{(0)}+x_{T+1}^{(0)}\right) \right)E_{T} \left(y_{T+1}^{(2)}\right).
\end{array}%
\end{equation}
 For $t=T-k+1$ we have 
\begin{equation}   \label{Ti2}
\begin{split}
&y_t^{(2)}=\frac{1}{2} \theta \sum _{i=1}^{k }\beta^{i} \exp \left(\theta
\sum _{j=1}^{i}x_{t+j}^{(0)} \right)\left[\theta
\big(1+y_{t+i}^{(0)}\big)
-2K_{\infty,t+i}\right] E_{t}\big(x_{t+i}^{(1)}\big)^{2} \\ +
&\beta^{k} exp\left( \theta \sum _{j=1}^{i}x_{t+j}^{(0)} \right) E_{t} y_{T+1}^{(2)}.
\end{split}
\end{equation}
If  the moment $T$ tends to $\infty$, then the following solution for $y_t^{(2)}$ is valid:
%
\begin{equation}  \label{66}
\begin{split}
y_{t}^{(2)} =\frac{1}{2} \theta \sum _{i=1}^{\infty }\beta ^{i} \exp \left(\theta
\sum _{j=1}^{i}x_{t+j}^{(0)} \right)\left[\theta
\big(1+y_{t+i}^{(0)}\big)
-2K_{\infty,t+i}\right] E_{t}\big(x_{t+i}^{(1)}\big)^{2} 
\end{split}
\end{equation}
At the time $t=0$ equation \eqref{66} provides the second term of the policy function series expansion 
\begin{equation}   
\begin{split}
&y_0^{(2)}=\frac{1}{2} \theta \sum _{i=1}^{\infty }\beta^{i} \exp \left(\theta
\sum _{j=1}^{i}x_{j}^{(0)} \right)\left[\theta
\big(1+y_{i}^{(0)}\big)
-2K_{\infty,i}\right] E_{0}\big(x_{i}^{(1)}\big)^{2} 
\end{split}
\end{equation}
The expectation term in the last equation can be obtained by using the moving-average representation for $x_{i}^{(1)}$. Indeed, from \eqref{62:2} and \eqref{59} we have 
\begin{equation*}  
x_{i}^{(1)} =\varepsilon _{i} +\rho \varepsilon _{i-1} +...+\rho
^{i-1} \varepsilon _{1}.
\end{equation*}
Since the sequence of innovations $\varepsilon_{i}$, $i>0$, 
is independent it follows that 
\begin{equation}  \label{68}
\begin{split}
E_{0}\left(x_{i}^{(1)} \right)^{2}&=E_{t}\left(\varepsilon _{i} +\rho \varepsilon _{i-1} +...+\rho
^{i-1} \varepsilon _{1}
\right)^{2}\\
&=1+\rho ^{2} +\cdots+\rho ^{2(i-1)} =\frac{1-\rho^{2i} }{1-\rho
^{2} }.
\end{split}
\end{equation}
The sum in exponential in \eqref{66} can be obtained from \eqref{57} 
\begin{equation}  \label{69}
\begin{array}{l}
x_{1}^{(0)} +x_{2}^{(0)} +\cdots+x_{i}^{(0)}=\bar{x}+\rho (x_{0}^{(0)} -\bar{x} )+\bar{x}+\rho^{2} (x_{0}^{(0)} -\bar{x} )\\
+\bar{x}+\rho ^{i} (x_{0}^{(0)} -\bar{x})= i\bar{x}+\frac{\rho (1-\rho ^{i} )}{1-\rho } (x_{0}^{(0)} -\bar{x}).
\end{array}%
\end{equation}
Finally, inserting 
\eqref{68} and \eqref{69} into \eqref{66} gives 
\begin{equation*}
y_{0}^{(2)} =\frac{\theta }{2} \sum _{i=1}^{\infty }\beta ^{i} \frac{1-\rho ^{2i}}{1-\rho^{2} }\exp \left\{\theta\Big[i\bar{x}+b (x_{0}^{(0)} -\bar{x})\Big]\right\}\left[\theta(1+y_{i}^{(0)})-2K_{\infty,i}\right].
\end{equation*}
In computation we need to use a finite terminal time $T+1$. Despite the fact that the method converges for any terminal condition $y_{T+1}^{(2)}$, the most reasonable choice of the terminal condition is the second order term in the  expansion of the stochastic steady state in a series of powers of $\sigma$. 
To summarize,  we find the policy function approximation in the form 
\begin{equation*}
y(x_0)=h(x_0)=y_{0}^{(0)} +\sigma ^{2} y_{0}^{(2)}.
\end{equation*}
Note that both $y_{0}^{(0)}$ and $y_{0}^{(2)}$ are functions of $x_0$.
From \eqref{66} and using \eqref{68}, we can get the expected dynamics (in other words, impulse response function) $E_{0}y_{t}$. 
The solutions for the higher orders $y_{t}^{(n)} (x)$, $n>2$, can be obtained in much the same way as for $y_{t}^{(2)}(x)$.

 \subsection{Comparison with the local perturbation}

This subsection compares the policy functions of  the second order of the
presented method with the local Taylor series expansions of orders two and six\citep{r14}. 

The parameterization follows \cite{r3}, where the benchmark parameterization is chosen as in \cite{r24}. 
We therefore set the mean of the rate of
growth of dividend to $\bar{x}=0.0179$, the volatility of the innovations to $\sigma $ = 0.015, the parameter $\theta $ to $-1.5$ 
and $\beta $ to $0.95$. 
For illustrative purpose,  we choose the highly persistent exogenous process with $\rho = 0.9$ as in \cite{r3}.
\begin{figure} 
\includegraphics [scale=0.65] {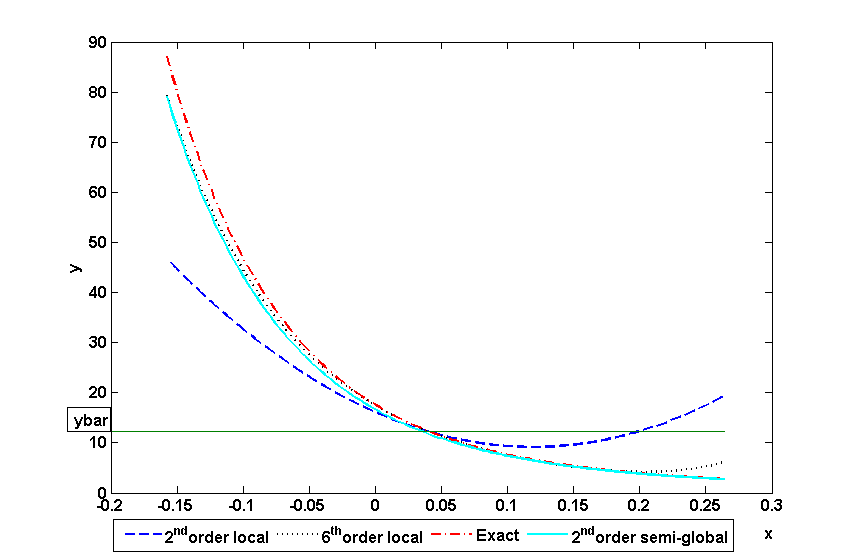}
\caption[]{ Comparison of the approximations of the policy function.  The level of the deterministic steady state of $y_t$ is denoted by ybar} 
\label{fig1}
\end{figure}
Fig~\ref{fig1}  illustrates  the exact policy function with the approximate ones constructed by  the semi-global method and the local Taylor series expansions. This figure is drawn over the
interval $x_{i} \in [\bar{x}- \Delta\cdot \sigma _{x} ,\bar{x}+\Delta \cdot
\sigma_{x} ]$, where $\sigma_{x}$ is the unconditional volatility of the
process $x_{t}$ and $\Delta=5$.
Fig~\ref{fig1}  shows that the semi-global approximation has the same accuracy as the sixth order of the Taylor series local expansion at the left endpoint of the interval under consideration. However, at the right endpoint  of the interval the semi-global solution is much more accurate than the sixth order of the Taylor series expansion. Actually, the semi-global approximation  is indistinguishable from the true solution in this domain. The second order of the Taylor series expansion is much less accurate globally than both  the sixth order of the Taylor series expansion and semi-global solution. 

From Fig~\ref{fig1} one can also see another undesirable property of the the local Taylor series expansion, namely this method may provide impulse response functions with  wrong signs. Indeed, the steady state value of $y_t$ is $\bar{y}=12.3$. After a big positive shock the true  impulse response function is negative (the policy function values are below the steady state), whereas the impulse response function implied by the second order of the local perturbation method is positive (the approximate policy function is above the steady state). The sixth order approximation of the local perturbation method has the right sign of impulse response, but wrong shape, which is U-shaped instead of being monotonically increasing. In contrast, the semi-global method, as just mentioned, provides almost exact impulse response. 

 \section{Conclusion}\label{conclusion}
This study proposes an approach based on a perturbation around a deterministic path for constructing  global approximate solutions to DSGE models. 
Under the assumption that the deterministic solution to the model is already found, the approach reduces the problem to solving recursively a set of  linear rational expectations models with deterministic  time-varying parameters and the same homogeneous part. 
The paper also proposes a method to solve linear rational expectations models with deterministic  time-varying parameters. 
The conditions under which the solutions exist are found; all results are obtained   for   DSGE models in general form and proved  rigorously.



The paper illustrates the algorithm for the second order of approximation using an nonlinear asset pricing model by  \cite{r16} and compares it with the local Taylor series expansion.
The second order approximation of the semi-global method provide more accurate solution than sixth order of the  Taylor series expansion around the deterministic steady state.

The approach is applicable to Markov-switching DSGE models in the form proposed by \cite{r37}, where the vector of Markov-switching parameters that would influence the steady state is scaled by a small factor. Actually, under the conditions of "smallness" of a scaling parameter and existence of higher order moments for stochastic terms, all derivations of Sections~\ref{expan}, \ref{transform} and \ref{sec6}  hold irrespective of probability distribution functions for these stochastic terms.  

%

\appendices
\section{Proofs for Section $5$}\label{A} 

 PROOF OF PROPOSITION~\ref{prop1}: The proof is by induction on $i$.
Suppose that $i=0$. For the time $T$ from \eqref{31}
we have 
\begin{equation*}
E_{T} u_{T+1} =B_{T} u_{T} +Q_{21,T} s_{T} +\Psi _{2,T} E_{T}
\eta _{T+1} .
\end{equation*}
As $B_{T}$ is invertible, we have 
\begin{equation*}
u_{T,T} =-K_{T,T} s_{T} - g_{T,0} +L_{T,T}^{-1}E_{T} u_{T+1},
\end{equation*}
where $K_{T,T}=B_{T}^{-1} Q_{21,T}$;  $g_{T,0} = -B_{T}^{-1} \Psi _{2,T}E_{T} \eta _{T+1}$ and $L_{T,T}^{-1} = B_{T}^{-1}$. 
From \eqref{32}, \eqref{33} and \eqref{36} it
follows that the inductive assumption is proved for $i=0$. Assuming that \eqref{35} holds for $i>0$, we will
prove it for $i+1$. To this end, consider Equation \eqref{31} for the time $t=T- i -1$. As the
matrix $B_{T-i}$ is invertible, we obtain 
\begin{equation*}
\begin{split}
u_{T,T-i-1} =-B_{T-i-1}^{-1} Q_{21,T-i-1} s_{T-i-1} - B_{T-i-1}^{-1} \Psi _{2,T-i-1}
E_{T-i-1} \eta _{T-i} \\
+ B_{T-i-1}^{-1} E_{T-i-1} u_{T,T-i}.
\end{split}
\end{equation*}
Substituting the induction assumption \eqref{35} for $u_{T,T-i}$ yields 
\begin{equation*}
\begin{array}{l}
 u_{T,T-i-1} = -B_{T-i-1}^{-1} Q_{21,T-i-1} s_{T-i-1} - B_{T-i-1}^{-1} \Psi _{2,T-i-1}
E_{T-i-1} \eta _{T-i} \\
+ B_{T-i-1}^{-1} E_{T-i-1}\left[-K_{T,T-i} s_{T-i}+g_{T,i} + \left(\prod _{k=1}^{i+1}L_{T,T-i+k}^{-1}\right) E_{T-i} \left(u_{T+1}\right)\right].
\end{array}
\end{equation*}
Substituting \eqref{30} for $E_{T-i-1}(s_{T-i})$ and using the law of iterated expectations gives 
\begin{equation*}
\begin{array}{l}
u_{T,T-i-1} =-B_{T-i}^{-1} Q_{21,T-i} s_{T-i-1} -B_{T-i}^{-1} \Psi_{2,T-i} E_{T-i-1} \eta _{T-i} +B_{T-i}^{-1} g_{T,i} \\
+B_{T-i}^{-1}\left(\prod _{k=1}^{i+1}L_{T,T-i+k}^{-1}\right) E_{T-i-1} \left(u_{T+1} \right)\\ 
+B_{T-i}^{-1}\left[-K_{T,T-i} \left(A_{T-i} s_{T-i-1} +Q_{12,T-i}
u_{T,T-i-1} +\Psi _{1,T-i} E_{T-i-1} \eta _{T-i}\right)\right] .
\end{array}%
\end{equation*}
Collecting the terms with $ u_{T,T-i-1}$, $s_{T-i-1}$ and $\eta_{T-i }$, we get 
\begin{equation*}  
\begin{array}{l}
\left(I+B_{T-i}^{-1} K_{T,T-i} Q_{12,T-i} \right)u_{T,T-i-1} = -B_{T-i}^{-1}\big[(Q_{21,T-i} +K_{T,T-i} A_{T-i})s_{T-i-1}  \\ 
+(\Psi _{2,T-i} +K_{T,T-i} \Psi _{1,T-i} )E_{T-i-1} \eta_{T-i}+ g_{T,i} + \left(\prod _{k=1}^{i+1}L_{T,T-i+k}^{-1}\right)
E_{T-i-1} \left(u_{T+1} \right)\big]\label{70}
\end{array}
\end{equation*}
Suppose for  the moment that the matrix
\[Z_{T,T-i} =I+B_{T-i}^{-1} K_{T,T-i} Q_{12,T-i}\]
is invertible. Pre-multiplying the last equation 
by $Z_{T,T-i}^{-1} $, we obtain 
\begin{equation*}
\begin{array}{l}
u_{T,T-i-1} =-Z_{T,T-i}^{-1}B_{T-i}^{-1}\big[(Q_{21,T-i} +K_{T,T-i} A_{T-i} )s_{T-i-1} \\ 
+(\Psi _{2,T-i} +K_{T,T-i} \Psi _{1,T-i} )E_{T-i-1} \eta_{T-i}+ g_{T,i} \\
+\left(\prod _{k=1}^{i+1}L_{T,T-i+k}^{-1}\right) E_{T-i-1} \left(u_{T+1} \right)\big].
\end{array}
\end{equation*}
Note that $L_{T,T-i} =B_{T-i} Z_{T,T-i}$; then using the definition of $K_{T,T-i-1}$ \eqref{32}, we see that 
\begin{equation} 
\begin{array}{l} 
u_{T,T-i-1} =-K_{T,T-i-1} s_{T-i-1}  \\
-L_{T,T-i}^{-1}\left(\Psi _{2,T-i} +K_{T,T-i} \Psi _{1,T-i}\right)E_{T-i-1} \eta_{T-i}\\
+ L_{T,T-i}^{-1}g_{T,i} + L_{T,T-i}^{-1}\left(\prod _{k=1}^{i+1}L_{T,T-i+k}^{-1}\right)E_{T-i-1}\left(u_{T+1} \right).\label{71}
\end{array}%
\end{equation}
Using the definition of $g_{T,i}$ and $L_{T-i,T-i+j} $ (\eqref{33} and 
\eqref{36}), we deduce that 
\begin{equation}  \label{72}
g_{T,i+1} =- L_{T,T-i}^{-1}\left(\Psi _{2,T-i} +K_{T,T-i} \Psi _{1,T-i}\right)E_{T-i-1} \eta_{T-i}\\
+ L_{T,T-i}^{-1}g_{T,i}.
\end{equation}
From \eqref{71} and \eqref{72} it follows that 
\begin{equation*}
u_{T,T-i-1} =-K_{T,T-i-1} s_{T-i-1} +g_{T,i+1} + \left(\prod _{k=1}^{i+2}L_{T,T-i-1+k}^{-1}\right) E_{T-i-1}
\left(u_{T+1} \right).
\end{equation*}
\qed

 PROOF OF PROPOSITION~\ref{prop2}: We begin by rewriting \eqref{32} as  
\begin{equation*}
\left(B_{T-i} +K_{T,T-i} Q_{12,T-i} \right)K_{T,T-(i+1)} =\left(Q_{21,T-i} +K_{T,T-i}A_{T-i}\right ).
\end{equation*}
Rearranging terms, we have 
\begin{equation}\label{72_}
\begin{split}
K_{T,T-(i+1)} =B_{T-i}^{-1} \cdot \left(Q_{21,T-i} +K_{T,T-i}A_{T-i}\right )\\
 - B_{T-i}^{-1}K_{T,T-i} Q_{12,T-i} K_{T,T-(i+1)} .
\end{split}
\end{equation}
Taking the norms and using the norm properties gives 
\begin{equation*}
\begin{split}
&\norm{ K_{T,T-(i+1)}}
   \leqslant \norm{ B_{T-i}^{-1}} \cdot \norm{
Q_{21,T-i} }+\norm{ B_{T-i}^{-1}}\cdot \norm{ K_{T,T-i}}
 \cdot \norm {A_{T-i} } \\
&+\norm{ B_{T-i}^{-1}} \cdot
\norm {K_{T,T-i} } \cdot \norm{ Q_{12,T-i}}\cdot \norm{
K_{T,T-(i+1)} }.
\end{split}
\end{equation*}
Rearranging terms, we get 
\begin{equation}
\lVert K_{T,T-(i+1)} \rVert \leqslant \frac{\lVert B_{T-i}^{-1} \rVert \cdot
\lVert Q_{21,T-i} \rVert +\lVert B_{T-i}^{-1} \rVert \cdot \lVert
K_{T,T-i} \rVert \cdot \lVert A_{T-i} \rVert }{1-\lVert B_{T-i}^{-1}
\rVert \cdot \lVert K_{T,T-i} \rVert \cdot \lVert Q_{12,T-i}
\rVert }.\label{72_1}
\end{equation}

Inequality \eqref{72_1} is a difference inequality with respect to $\lVert K_{T,T-i} \rVert$, $i=0,1,\ldots,T,$ and 
with the time-varying coefficients 
$\lVert A_{T-i} \rVert$, $\lVert B_{T-i}^{-1} \rVert$, $\lVert Q_{12,T-i}
\rVert$ and $\lVert Q_{21,T-i} \rVert $. In \eqref{72_1} we assume that 

$1-\lVert B_{T-i}^{-1} \rVert \cdot \lVert
K_{T,T-i} \rVert \cdot \lVert Q_{12,T-i} \rVert \ne 0$.

This is obviously true if $\lVert K_{T,T-i} \rVert =0$. We shall show that if the initial condition $\lVert K_{T,T+1} \rVert =0$, then 
$\left(1-\lVert B_{T-i}^{-1} \rVert \cdot \lVert
K_{T,T-i} \rVert \cdot \lVert Q_{12,T-i} \rVert\right) > 0$, $i=1, 2,\ldots,T.$  Indeed, consider the difference equation: 
\begin{equation}  \label{73}
s_{i+1} =\frac{bd+bas_{i} }{\left(1-bcs_{i} \right)}.
\end{equation}
\begin{lemma}\label{lem1}
 If  inequality \eqref{38}
holds, then the difference equation \eqref{73} has two fixed points 
\begin{equation}  \label{75}
s_{1}^{*} =\frac{2bd}{1-ba+\sqrt{(1-ba)^{2} -4b^{2} cd} },
\end{equation}
\begin{equation*}
s_{2}^{*} =\frac{1-ba+\sqrt{(1-ba)^{2} -4b^{2} cd} }{2bc}, 
\end{equation*}
 where $s_{1}^{*} $ is a stable fixed point whereas $s_{2}^{*} $ is
an unstable one. Moreover, under the initial condition $s_{0} =0$ the solution $s_{i}, i = 1,2,\dots,$ is an  increasing sequence
and converges to $s_{1}^{*} $.
\end{lemma}
The lemma can be proved by direct calculation.
From \eqref{27}--\eqref{27_0} the values $a$, $b$, $c$ and $d$ majorize $\lVert A_{T-i} \rVert$ , $\lVert B_{T-i}^{-1} \rVert$, $\lVert Q_{12,T-i}\rVert$ and $\lVert Q_{21,T-i} \rVert$, respectively.
If we consider Equation \eqref{72} and inequality \eqref{73} as initial value problems with the initial conditions $\lVert K_{T,T+1} \rVert =0$ and $s_{0} =0$, then their solutions obviously satisfy the inequality $\lVert K_{T,T-i} \rVert \leqslant s_{i+1} $, $i=1,2, \dots , T$.
%
In other words, $\lVert K_{T,T-i} \rVert ${\ } is majorized by $s_{i}$.
From the last inequality 
 and Lemma~\ref{lem1} it may be concluded that
\begin{equation} \label{76}
 \lVert K_{T,T-i} \rVert \leqslant s_{1}^{*},\quad i =0,1,2,\ldots, T, \quad T\in \mathbb{N}.
\end{equation}

From \eqref{75}, \eqref{76} and %
\eqref{27} it follows that  
\begin{equation}\label{76_}
\lVert B_{T-i}^{-1} \rVert \cdot \lVert K_{T,T-i} \rVert \cdot
\lVert Q_{12,T-i} \rVert \leqslant \frac{2b^{2} dc}{1-ba+\sqrt{(1-ba)^{2}
-4b^{2} cd} } . 
\end{equation}
From \eqref{38} we see that $2b^{2} dc<(1-ab)^{2}/2$. Substituting this inequality into \eqref{76_} gives
\begin{equation}\label{76_1}
\begin{split}
\lVert B_{T-i}^{-1} \rVert \cdot \lVert K_{T,T-i} \rVert \cdot
\lVert Q_{12,T-i} \rVert &\leqslant \frac{\left(1-ba\right)^{2} }{2(1-ba+\sqrt{%
(1-ba)^{2} -4b^{2} cd)} }\\
&<\frac{\left(1-ba\right)^{2} }{2(1-ba)} =\frac{1-ba}{2} <1,
\end{split}
\end{equation}
where the last inequality follows from \eqref{29}. 
\qed

 PROOF OF PROPOSITION~\ref{prop4}: The assertion of the proposition is true
if there exist constants $M$ and $ r$ such that $0<r<1$ and for $T\in \mathbb{N}$ 
\begin{equation}  \label{77}
\lVert K_{T,j} -K_{T+1,j} \rVert \leqslant Mr^{T+1},\quad j=0,1,2,\ldots.
\end{equation}
Note now that $K_{T,j} $ ($K_{T+1,j} $) is a solution to the matrix
difference equation \eqref{32} at $i=T-j$ ($i=T+1-j$) with
the initial condition $K_{T,T+1} =0$ ($K_{T+1,T+2} =0$). 
Subtracting \eqref{72_} for $K_{T,T-(i+1)}$ from that for $K_{T+1,T-(i+1)}$, we have
\begin{equation*}
\begin{split}
&K_{T,T-(i+1)} -K_{T+1,T-(i+1)}=B_{T-i}^{-1}  (K_{T,T-i)} -K_{T+1,T-i})A_{T-i}\\
&-B_{T-i}^{-1}K_{T,T-i)} Q_{12,T-i} K_{T,T-(i+1)} + B_{T-i}^{-1}
K_{T+1,T-i} Q_{12,T-i} K_{T+1,T-(i+1)}.
\end{split}
\end{equation*}
Adding and subtracting $B_{T-i}^{-1}\cdot K_{T,T-i} \cdot Q_{12,T-i} \cdot K_{T+1,T-(i+1)}$ in the right hand side gives
\begin{equation*}
\begin{split}
&K_{T,T-(i+1)} -K_{T+1,T-(i+1)}=B_{T-i}^{-1}  (K_{T,T-i)} -K_{T+1,T-i})A_{T-i}\\
&-B_{T-i}^{-1}\cdot K_{T,T-i} \cdot Q_{12,T-i} (K_{T,T-(i+1)}-K_{T+1,T-(i+1)}) \\
&-B_{T-i}^{-1} (K_{T,T-i} - K_{T+1,T-i} )Q_{12,T-i}\cdot  K_{T+1,T-(i+1)}.
\end{split}
\end{equation*}
Rearranging terms yields
\begin{equation*}
\begin{split}
&(I+B_{T-i}^{-1}K_{T,T-i}  Q_{12,T-i} )(K_{T,T-(i+1)} -K_{T+1,T-(i+1)})\\
&=B_{T-i}^{-1} (K_{T,T-i} -K_{T+1,T-i})A_{T-i}\\
& -B_{T-i}^{-1}(K_{T,T-i}- K_{T+1,T-i} )Q_{12,T-i} K_{T+1,T-(i+1)}.
\end{split}
\end{equation*}
From Proposition~\ref{prop3} it follows that the matrix 
\[Z_{T,T-i}=(I+B_{T-i}^{-1}K_{T,T-i}  Q_{12,T-i} )\]
is invertible, then pre-multiplying the last equation by this matrix yields
\begin{equation*}
\begin{split}
&K_{T,T-(i+1)} -K_{T+1,T-(i+1)}=Z_{T,T-i}^{-1}(B_{T-i}^{-1} (K_{T,T-i} -K_{T+1,T-i})A_{T-i}\\
& -B_{T-i}^{-1}(K_{T,T-i)}- K_{T+1,T-i} )Q_{12,T-i} K_{T+1,T-(i+1)}).
\end{split}
\end{equation*}
Taking the norms, using the norm property and the triangle inequality, we get
\begin{equation}
\begin{split}
&\lVert K_{T,T-(i+1)} -K_{T+1,T-(i+1)}\rVert\\
&\leqslant \lVert Z_{T,T-i}^{-1}\rVert \cdot (\lVert B_{T-i}^{-1}\rVert \cdot \lVert K_{T,T-i} -K_{T+1,T-i}\rVert \cdot \lVert A_{T-i}\rVert\\
& + \lVert B_{T-i}^{-1}\rVert \cdot \lVert K_{T,T-i)}- K_{T+1,T-i} \rVert \cdot  \lVert Q_{12,T-i}\rVert \cdot  \lVert K_{T+1,T-(i+1)}\rVert ).
\end{split}
\end{equation}

From \eqref{27_0} and \eqref{76_1} we have
\begin{equation}\label{78}
\begin{split}
&\lVert K_{T,T-(i+1)} -K_{T+1,T-(i+1)}\rVert\\
&\leqslant \left(ab + \frac{1-ba}{2}\right)  \lVert Z_{T,T-i}^{-1}\rVert\cdot \lVert K_{T,T-i} -K_{T+1,T-i}\rVert \\
&=\frac{1+ba}{2}\lVert Z_{T,T-i}^{-1}\rVert \cdot   \lVert K_{T,T-i} -K_{T+1,T-i}\rVert .
\end{split}
\end{equation}
From  the norm property and \citet[Lemma 2.3.3]{r26}   we get the estimate 
\begin{equation*}
\begin{split}
\lVert Z_{T,T-i}^{-1}\rVert&=\lVert(I+B_{T-i}^{-1}K_{T,T-i}  Q_{12,T-i} )^{-1}\rVert \leqslant \frac{1}
{1-\lVert B_{T-i}^{-1}K_{T,T-i}  Q_{12,T-i}  \rVert} \\
&\leqslant \frac{1}{1-\lVert B_{T-i}^{-1}\rVert \cdot \lVert K_{T,T-i}\rVert \cdot \lVert Q_{12,T-i}  \rVert}
\end{split}
\end{equation*}
By \eqref{76_1}, we have
\begin{equation*}
\lVert Z_{T,T-i}^{-1}\rVert=
< \frac{1}
{1-\frac{1-ba}{2}} = \frac{2}
{1+ba}
\end{equation*}
Substituting the last inequality into \eqref{78} gives
\begin{equation}\label{79}
\begin{split}
&\lVert K_{T,T-(i+1)} -K_{T+1,T-(i+1)}\rVert<   \lVert K_{T,T-i} - K_{T+1,T-i}\rVert .
\end{split}
\end{equation}
Using \eqref{79} successively for $i=-1,0,1,\ldots, T-1$, and taking into account $K_{T,T+1} =0$ and $K_{T+1,T+1} =B^{-1}_{T+2}Q_{21,T+2}$ results in
\begin{equation}\label{79}
\begin{split}
&\lVert K_{T,j} -K_{T+1,j}\rVert< \lVert K_{T,T+1} -K_{T+1,T+1}\rVert=  \lVert B_{T+2}^{-1} Q_{21,T+2}\rVert \\
& \leqslant \lVert B_{T+2}^{-1}\rVert \cdot \lVert Q_{21,T+2}\rVert \leqslant b\lVert Q_{21,T+2}\rVert,\quad j=0,1,2,\ldots.
\end{split}
\end{equation}
Recall that $Q_{21,T}$ depends on the solution to the deterministic problem \eqref{14}, i.e. \[Q_{21,T}=Q_{21}\left(x_{T+1}^{(0)} ,x_{T}^{(0)} ,z_{T+1}^{(0)} ,z_{T}^{(0)}\right).\]
From \citet[Corollary 5.1]{r6} 
and differentiability of $Q_{21}$ with respect to the state variables it follows that 
 \begin{equation}\label{80}
\lVert Q_{21,T}\rVert \leqslant C(\alpha+\theta)^{T}, 
\end{equation}
where $\alpha$ is the largest eigenvalue modulus of the matrix $A$ from \eqref{25_}, $C$ is some constant and $\theta$ is arbitrary small positive number. In fact, $\alpha+\theta$ determines the speed of convergence for the deterministic solution to the steady state. 
Inserting \eqref{80} into \eqref{79}, we can conclude
\begin{equation}\label{79}
\begin{split}
&\lVert K_{T,j} -K_{T+1,j}\rVert<  bC(\alpha+\theta)^{T+2},\quad j=0,1,2,\ldots
\end{split}
\end{equation}
Denoting $M=bC(\alpha+\theta)$ and $r=\alpha+\theta$ we finally obtain \eqref{77}. 
\qed

\section{The First Order System}\label{FOS2}
For $n=1$ we have
 \[\left[%
\begin{array}{c}
{s_{0}^{(1)} } \\ 
{ u_{0}^{(1)}}%
\end{array}%
\right]=Z_{1}Z^{-1} \left[%
\begin{array}{c}
{x_{0}^{(1)} } \\ 
{y_{0}^{(1)} }%
\end{array}%
\right] =0,\]
 From \eqref{26_} for the time $T$ 
we have
\begin{equation*}\label{1:1}
u_{T}^{(1)} =-B_{T+1}^{-1} Q_{21,T+1} s_{T}^{(1)} -B_{T+1}^{-1} {\Pi}_{2,t+1}z_{T}^{(1)}
+B_{T+1}^{-1} E_{T} u_{T+1}^{(1)} .
\end{equation*}
Denoting $K_{T,T} =B_{T+1}^{-1} Q_{21,T+1}$ and $R_{T} = B_{T+1}^{-1} {\Pi}_{2,t+1}$ gives
\begin{equation}\label{1:1_}
u_{T}^{(1)} =-K_{T,T}  s_{T}^{(1)} - R_{T} z_{T}^{(1)}
+B_{T+1}^{-1} E_{T} u_{T+1}^{(1)} .
\end{equation}
For $T-1$ we have 
\begin{equation}\label{1:2}
u_{T-1}^{(1)} =-B_{T}^{-1} Q_{21,T} s_{T-1}^{(1)}  -B_{T}^{-1} {\Pi}_{2,t}z_{T-1}^{(1)}
+B_{T}^{-1} E_{T-1} u_{T}^{(1)} .
\end{equation}
Taking conditional expectations at the time $T-1$ from both side \eqref{1:1_}  
and inserting \eqref{2} we get
\begin{equation}\label{1:3_}
E_{T-1}u_{T}^{(1)} =-K_{T,T} E_{T-1}s_{T}^{(1)}  - R_{T}\Lambda z_{T-1}^{(1)}
+B_{T+1}^{-1} E_{T-1} u_{T+1}^{(1)} .
\end{equation}

Inserting \eqref{1:3_} 
into \eqref{1:2} gives
\begin{equation}\label{1:4}
\begin{split}
&u_{T-1}^{(1)} =-B_{T}^{-1} Q_{21,T} s_{T-1}^{(1)}  - B_{T}^{-1} {\Pi}_{2,t}z_{T-1}^{(1)}\\
&+B_{T}^{-1} E_{T-1} (-K_{T,T}E_{T-1}s_{T}  - R_{T}\Lambda z_{T-1}^{(1)}
+B_{T+1}^{-1} E_{T} u_{T+1}^{(1)}) .
\end{split}
\end{equation}

%
%

Inserting now $ E_{T-1}s_{T}$ into \eqref{1:4} from \eqref{26} yields
\begin{equation*}\label{1:5_}
\begin{split}
&u_{T-1}^{(1)} =-B_{T}^{-1} Q_{21,T} s_{T-1}^{(1)}  - B_{T}^{-1} {\Pi}_{2,t}z_{T-1}^{(1)}\\
&+B_{T}^{-1} E_{T-1} [-K_{T,T}(A_{T-1}s_{T-1}^{(1)}  +Q_{12,T} u_{T-1}^{(1)} +{\Pi}_{1,T}z_{T-1}^{(1)})\\
&- R_{T}\Lambda z_{T-1}^{(1)}
+B_{T+1}^{-1} E_{T} u_{T+1}^{(1)}] .
\end{split}
\end{equation*}

Reshuffling terms, we have
\begin{equation}\label{1:5_2}
\begin{split}
&(I +B_{T}^{-1} K_{T,T}Q_{12,T} ) u_{T-1}^{(1)} =-B_{T}^{-1}( Q_{21,T}+ B_{T}^{-1} K_{T,T}A_{T-1})s_{T-1}^{(1)} \\
&- B_{T}^{-1}( {\Pi}_{2,t}+  K_{T,T}{\Pi}_{1,T}+R_{T}\Lambda) z_{T-1}^{(1)} +B_{T}^{-1}B_{T+1}^{-1} E_{T} u_{T+1}^{(1)} .
\end{split}
\end{equation}
%
Multiplying \eqref{1:5_2} by $(I +B_{T}^{-1} K_{T,T}Q_{12,T} )^{-1}$ yields
\begin{equation*}\label{1:6-1}
\begin{split}
&u_{T-1}^{(1)} =-(I +B_{T}^{-1} K_{T,T}Q_{12,T} )^{-1}B_{T}^{-1}( Q_{21,T}+ B_{T}^{-1} K_{T,T}A_{T-1})s_{T-1}^{(1)} \\
&- (I +B_{T}^{-1} K_{T,T}Q_{12,T} )^{-1}B_{T}^{-1}( {\Pi}_{2,t}+  K_{T,T}{\Pi}_{1,T}+R_{T}\Lambda) z_{T-1}^{(1)} \\
&+(I +B_{T}^{-1} K_{T,T}Q_{12,T} )^{-1}B_{T}^{-1}B_{T+1}^{-1} E_{T} u_{T+1}^{(1)} .
\end{split}
\end{equation*}
or
\begin{equation*}\label{1:6_3}
\begin{split}
&u_{T-1}^{(1)} =-(B_{T}+K_{T,T}Q_{12,T} )^{-1}( Q_{21,T}+ B_{T}^{-1} K_{T,T}A_{T-1})s_{T-1}^{(1)} \\
&- (B_{T}+K_{T,T}Q_{12,T} )^{-1}( {\Pi}_{2,t}+  K_{T,T}{\Pi}_{1,T}+R_{T}\Lambda) z_{T-1}^{(1)} \\
&+(B_{T}+K_{T,T}Q_{12,T} )^{-1}B_{T+1}^{-1} E_{T} u_{T+1}^{(1)}.
\end{split}
\end{equation*}
Denoting $L_{T,T-1} = (B_{T-1}+K_{T,T}Q_{12,T-1} )$, we obtain
\begin{equation*}\label{1:6_3}
\begin{split}
&u_{T-1}^{(1)} =-L_{T,T-1}^{-1}( Q_{21,T}+ B_{T}^{-1} K_{T,T}A_{T-1})s_{T-1}^{(1)} \\
&-L_{T,T-1}^{-1}( {\Pi}_{2,t}+  K_{T,T}{\Pi}_{1,T}+R_{T}\Lambda) z_{T-1}^{(1)} \\
&+L_{T,T-1}^{-1}B_{T+1}^{-1} E_{T} u_{T+1}^{(1)}.
\end{split}
\end{equation*}
Denoting 
\[K_{T,T-1} = L_{T,T-1}^{-1}( Q_{21,T}+ B_{T}^{-1} K_{T,T}A_{T-1}) \]
 and 
\[R_{T-1} = L_{T,T-1}^{-1}( {\Pi}_{2,t}+  K_{T,T}{\Pi}_{1,T}+R_{T}\Lambda), \]
we have
\begin{equation*}\label{1:6_4}
\begin{split}
u_{T-1}^{(1)} =-K_{T,T-1} s_{T-1}^{(1)} 
-R_{T-1} z_{T-1}^{(1)} 
+L_{T,T-1}^{-1}L_{T,T}^{-1} E_{T} u_{T+1}^{(1)}.
\end{split}
\end{equation*}

Following the same derivation as in \ref{A} for the  proof of Proposition~\ref{prop1}, we obtain the following representation:
\begin{equation}\label{1:6}
u_t^{(1)} = - K_{T,t}s_t^{(1)}- R_t z_t^{(1)}  ,
\end{equation}
where $R_t$ can be computed by backward recursion 
\begin{equation*}
R_t = L^{-1}_{T,t+1}(\Pi_{2,t} + K_{T,t}\Pi_{1,t} + R_{t+1}\Lambda)
\end{equation*}
Inserting \eqref{26} into \eqref{1:6} gives
\begin{equation*}
E_{t} s_{t+1}^{(1)} =As_{t}^{(1)} +Q_{11,t+1} s_{t}^{(1)} +Q_{12,t+1} ( - K_{T,t}s_t^{(1)} - R_t z_t^{(1)})+{\Pi}_{1,t+1}z_{t}^{(1)}
\end{equation*}
After reshuffling we get
\begin{equation*}
E_{t} s_{t+1}^{(1)} = (A_{t+1}-Q_{12,t+1}K_{T,t})s_{t}^{(1)} +(-Q_{12,t+1} R_t +{\Pi}_{1,t+1})z_{t}^{(1)}.
\end{equation*}
Denoting $\mathbb{A}_{t} = A_{t+1} - Q_{12,t+1}K_{T ,t} $ and $\mathbb{P}_{t} = -  Q_{12,t+1} R_t +{\Pi}_{1,t+1}$, we have
\begin{equation}\label{1:7_}
E_{t} s_{t+1}^{(1)} = \mathbb{A}_{t}s_{t}^{(1)} +\mathbb{P}_{t} z_{t}^{(1)}
\end{equation}
It is easy to see that
\begin{equation*}
\left[%
\begin{array}{c}
{s_{t+1}^{(1)} } \\ 
{u_{t+1}^{(1)} }%
\end{array}%
\right] - E_{t} \left[%
\begin{array}{c}
{s_{t+1}^{(1)} } \\ 
{u_{t+1}^{(1)} }%
\end{array}%
\right]=\left[%
\begin{array}{c}
\mathbb{R}_{1,t} \\ 
\mathbb{R}_{2,t} 
\end{array}
\right] \varepsilon _{t+1}=Z\Phi _{t+1}^{-1} f_{5 ,t+1}\varepsilon _{t+1}.
\end{equation*}\label{1:9}
From \eqref{1:7_} it follows that
\begin{equation*}
(E_{t}s_{t+1} - s_{t+1}) + s_{t+1} = \mathbb{A}_{t}s_{t}+\mathbb{P}_{t}z_t ,
\end{equation*}
thus,  we obtain 
\begin{equation}\label{stable}
 s_{t+1} = \mathbb{A}_{t}s_{t}+\mathbb{P}_{t}z_t + \mathbb{R}_{1,t}\varepsilon _{t+1}.
\end{equation} 
Recall now that the initial conditions are $s_0^{(1)}=0$ and $z_{0}^{(1)}=0$, then for $t=1$ from \eqref{stable} we have
\[ s_1^{(1)}= \mathbb{R}_{1,0}\varepsilon _{1};\] for $t=2$ 
\[   s_{2}^{(1)} = (\mathbb{A}_{1}\mathbb{R}_{1,0}+\mathbb{P}_{t}) \varepsilon _{1}+\mathbb{R}_{1,1}\varepsilon _{2}. \]
Continuing in this fashion, we get the moving-average representation of $s_{t}^{(1)} $:
\begin{equation}\label{1:10}
s_{t}^{(1)}  = \gamma_{t,t}\varepsilon _{t} + \gamma_{t,t-1}\varepsilon _{t-1}+\cdots+ \gamma_{t,2}\varepsilon _{2} +\gamma_{t,1}\varepsilon _{1},
\end{equation}
where the coefficients $\gamma_{t,t-i}$ can be obtained by forward recursion in $t=1,2,\dots,T$ and backward recursion in $i=0,1,\dots,t-1$ 
\begin{align*}
\gamma_{t,t} = \mathbb{R}_{1,t-1},\\
\gamma_{t,t-1} = \mathbb{A}_{t-1}\gamma_{t-1,t-1}+\mathbb{P}_{t-1},\\
\ldots\\
\gamma_{t,t-i} = \mathbb{A}_{t-1}\gamma_{t-1,t-i}+\mathbb{P}_{t-1}{\Lambda}^{i-1},\\
\ldots\\
\gamma_{t,1} = \mathbb{A}_{t-1}\gamma_{t-1,1}+\mathbb{P}_{t-1}{\Lambda}^{t-2}
\end{align*}
Indeed, inserting \eqref{1:10} into \eqref{stable} and taking into account $z_{t} = \varepsilon _{t} + \Lambda \varepsilon _{t-1}+\cdots+ \Lambda^{t-1}\varepsilon _{1}$, we obtain
\begin{equation}
\begin{split}
 s_{t+1} = \mathbb{A}_{t}( \gamma_{t,t}\varepsilon _{t} + \gamma_{t,t-1}\varepsilon _{t-1}+\cdots+ \gamma_{t,2}\varepsilon _{2} +\gamma_{t,1}\varepsilon _{1})\\
+\mathbb{P}_{t}( \varepsilon _{t} + \Lambda \varepsilon _{t-1}+\cdots+ \Lambda^{t-1}\varepsilon _{1}) + \mathbb{R}_{1,t}\varepsilon _{t+1},
\end{split}
\end{equation} 
Collecting terms with $\varepsilon _{j}$ gives
\begin{equation*}
\begin{split}
 s_{t+1} = \mathbb{R}_{1,t}\varepsilon _{t+1}+ (\mathbb{A}_{t} \gamma_{t,t} + \mathbb{P}_{t})\varepsilon _{t} \\
+ (\mathbb{A}_{t} \gamma_{t,t-1} 
+ \mathbb{P}_{t}\Lambda)\varepsilon _{t-1}
+\cdots+  (\mathbb{A}_{t} \gamma_{t,1} + \mathbb{P}_{t}\Lambda^{t-1})\varepsilon _{1}.
\end{split}
\end{equation*}

Thus, for each $t$ we compute $\gamma_{t,i}$, starting with the first index $t=1$, then decreasing the index $i=t,t-1,\dots,1$ and using at each step $\gamma_{t-1,i}$.
%
 For the variable $u^{(1)}_t$ we also have a moving-average representation.
Inserting the moving-average representation of the process $z^{(1)}_{t}$ and \eqref{1:10} in \eqref{1:6}, we have
\begin{equation}
u^{(1)}_t =- K_{T,t}(\gamma_{t,t}\varepsilon _{t} +\cdots+ \gamma_{t,1}\varepsilon _{1}) - R_t(\varepsilon _{t} + \Lambda\varepsilon _{t-1}+\dots+{\Lambda}^{t-1}\varepsilon _{1}),
\end{equation}
or in the shorter form
\begin{equation}
u^{(1)}_t = \delta_{t,t}\varepsilon _{t} +\delta_{t,t-1}\varepsilon _{t-1}+\cdots+ \delta_{t,2}\varepsilon _{2} +\delta_{t,1}\varepsilon _{1},
\end{equation}
where $\delta_{t,i} =- K_{T,t}\gamma_{t,i} - R_{t}{\Lambda}^{i-1}$. 

Taking into account that $x^{(1)}_{t } = Z_{11}s^{(1)}_t + Z_{12}u^{(1)}_t $ and $y^{(1)}_{t } = Z_{21}s^{(1)}_t + Z_{22}u^{(1)}_t $, we get 
the moving-average representation for original variables 
 \[x^{(1)}_{t } = \rho_{t,t}^{x}\varepsilon _{t} +\rho_{t,t-1}^{x}\varepsilon _{t-1}+\cdots+ \rho_{t,2}^{x}\varepsilon _{2} +\rho_{t,1}^{x}\varepsilon _{1},\] 
\[y^{(1)}_{t } = \rho_{t,t}^{y}\varepsilon _{t} +\rho_{t,t-1}^{y}\varepsilon _{t-1}+\cdots+ \rho_{t,2}^{y}\varepsilon _{2} +\rho_{t,1}^{y}\varepsilon _{1},\]
where $  \rho_{t,i}^{x} =  Z_{11}\gamma_{t,i} +Z_{12}\delta_{t,i}$ and  $ \rho_{t,i}^{y} =  Z_{21}\gamma_{t,i} +Z_{22}\delta_{t,i}$.

\bibliographystyle{elsarticle-harv}

\begin{thebibliography}{00}

\bibitem[Abraham, Marsden, and Ratiu  (2001)]{r22}
\textsc{Abraham, R.}, \textsc{J.E. Marsden}, and \textsc{T. Ratiu}  (2001): \textit{Manifolds, Tensor Analysis, and Applications}, 2nd ed.
Springer-Verlag, Berlin-Heidelberg-New-York-Tokyo. 
\bibitem[Adjemian, Bastani, Juillard, Karam\'{e}, Mihoubi, Perendial, Pfeifer, Ratto, and Villemot   (2011)]{r0}
\textsc{Adjemian, S.}, \textsc{H. Bastani}, \textsc{M. Juillard}, \textsc{F. Karam\'{e}}, \textsc{F. Mihoubi}, \textsc{G. Perendia}, \textsc{J. Pfeifer}, \textsc{M. Ratto}, and \textsc{S. Villemot}  (2011):
``Dynare: Reference Manual, Version 4." Dynare Working Papers, 1, CEPREMA
\bibitem[Anderson and Moor  (1985)]{r30}
\textsc{Anderson, G.}, and \textsc{G. Moor} (1985):
``A Linear Algebraic Procedure for Solving Linear Perfect Foresight Models," \textit{Economics Letters}
\textbf{17} 247--252

\bibitem[\protect\citeauthoryear{Andreasen, Fern\'{a}ndez-Villaverde, and Rubio-Ramírez}{2013}]{r1}
\textsc{Andreasen, M.}, \textsc{J. Fern\'{a}ndez-Villaverde}, and \textsc{J. F. Rubio-Ramírez}(2013):
``The Pruned State-Space System for Non-Linear DSGE Models: Theory and Empirical Applications," NBER Working Paper No. 18983.

\bibitem[Blanchard, and Kahn (1980)]{r25}
\textsc{Blanchard, O.J.}, and \textsc{C.M. Kahn}  (1980):
``The Solution of Linear Difference Models Under Rational Expectations," \textit{ Econometrica}
\textbf{48} 1305--1311.
\bibitem[Burnside  (1998)]{r16}
\textsc{Burnside, C.}  (1998):
"Solving asset pricing models with Gaussian shocks," \textit{Journal of Economic Dynamics and Control }
\textbf{22} 329--340
\bibitem[Collard, and Juillard  (2001)]{r3}
\textsc{Collard, F.}, and \textsc{M. Juillard} (2001):
``Accuracy of stochastic perturbation methods: the case of asset pricing models," \textit{Journal of Economic Dynamics and Control }
\textbf{25} 979--999.


\bibitem[Fair, and Taylor (1983)]{r15}
\textsc{Fair, R.}, and \textsc{J.Taylor}  (1983):
``Solution and maximum likelihood estimation of dynamic rational expectation models," \textit{Econometrica}
\textbf{51} 1169--1185.

\bibitem[\protect\citeauthoryear{Foerster, Rubio-Ramíres, Waggoner and  Zha}{2013}]{r37}
\textsc{Foerster, A.}, \textsc{J. Rubio-Ramíres}, \textsc{D. Waggoner}, and \textsc{T. Zha} (2013):
``Perturbation Methods for Markov-Switching DSGE Model." Federal Reserve Bank of Kansas City Research Working Paper, No. RWP 13-01. 

\bibitem[Gaspar, and Judd  (1997)]{r17}
\textsc{Gaspar, J.} ,and \textsc{K. L. Judd }  (1997):
``Solving Large-Scale Rational-Expectations Models," \textit{Macroeconomic Dynamics}
\textbf{1} 45--75.
\bibitem[Golub, and Van Loan  (1996)]{r26}
\textsc{Golub, G.H.}, and {C.F. Van Loan}  (1996): \textit{Matrix Computations}, 3rd ed.
Johns Hopkins University Press, Baltimore.
\bibitem[Gomme, and Klein (2011)]{r5}
\textsc{Gomme, P.}, and \textsc{P. Klein}  (2011):
``Second-Order Approximation of Dynamic Models without the Use of Tensors," \textit{Journal of Economic Dynamics and Control }
\textbf{35} 604--615.
\bibitem[Hartmann (1982)]{r6}
\textsc{Hartmann, P.}  (1982): \textit{Ordinary Differential Equations}, 2nd ed.
Wiley, New York.
\bibitem[Hollinger (2008)]{r36}
\textsc{Hollinger, P.} (2008):
``How TROLL Solves a Million Equations: Sparse-Matrix Techniques for Stacked-Time Solution of Perfect-Foresight Models", presented at the 14th International Conference on Computing in Economics and Finance, Paris, France, 26–28 June, 2008, Intex Solutions, Inc., Needham, MA, http://www.intex.com/troll/Hollinger\_CEF2008.pdf
\bibitem[Holmes (2013)]{r20}
\textsc{Holmes, M. H.}  (2013): \textit{Introduction to Perturbation Methods}, 3rd ed
Springer-Verlag, Berlin-Heidelberg-New-York-Tokyo. 
\bibitem[Jin, and Judd  (2002)]{r7}
\textsc{Jin, H.}, and \textsc{K. L. Judd} (2002):
``Perturbation methods for general dynamic stochastic models." Discussion Paper, Hoover Institution, Stanford.
\bibitem[Judd (1998)]{r8}
\textsc{Judd, K. L.}  (1998): \textit{Numerical Methods in Economics}, 
The MIT Press, Cambridge.


\bibitem[Judd, and Guu (1997)]{r18}
\textsc{Judd K. L.}, and \textsc{S.-M. Guu } (1997):
"Asymptotic Methods for Aggregate Growth Models," \textit{Journal of Economic Dynamics and Control}
\textbf{21} 1025--1042.



\bibitem[Juillard (1996)]{r23}
\textsc{Juillard, M.}  (1996):
``DYNARE: a program for the resolution and simulation of dynamic models with forward variables through the use of a relaxation algorithm." CEPREMAP working paper No. 9602, Paris
\bibitem[Kim et al. (2008)]{r10} 
\textsc{Kim, J. et al.}, \textsc{S.Kim}, \textsc{E.Schaumburg}, and \textsc{C. A.Sims} (2008):
``Calculating and using second order accurate solutions of discrete time dynamic equilibrium models," \textit{Journal of Economic Dynamics and Control}
\textbf{32} 3397--3414.
\bibitem[Klein (2000)]{r5_}
\textsc{Klein P. }  (2000):
``Using the generalized Schur form to solve a multivariate linear rational expectations model," \textit{Journal of Economic Dynamics and Control }
\textbf{35} 1405--1423.
\bibitem[ Lombardo (2010)]{r19}
\textsc{Lombardo, G.} (2010):
``On Approximating DSGE Models by Series Expansions," European Central Bank Working Paper Series, No. 1264.
\bibitem[ Lombardo, and Uhlig (2014)]{r19_}
\textsc{Lombardo, G.}, and \textsc{H. Uhlig} (2014):
``A Theory of Pruning," European Central Bank Working Paper Series, No. 1696.
\bibitem[Mehra, and Prescott (1985)]{r24}
\textsc{Mehra, R.}, and \textsc{ E.C. Prescott} (1985):
``The Equity Premium: a Puzzle," \textit{Journal of Monetary Economics}
\textbf{15} 145--161.

\bibitem[Nayfeh (1973)]{r88}
\textsc{Nayfeh, A. H.}  (1973): \textit{ Perturbation Methods}, 
Wiley, New York.


\bibitem[Schmitt-Groh\'{e}, and Uribe (2004)]{r14}
\textsc{Schmitt-Groh\'{e}, S.}, and \textsc{ M.Uribe}  (2004):
``Solving dynamic general equilibrium models using as second-order approximation to the policy function," \textit{Journal of Economic Dynamics and Control}
\textbf{28} 755--775.
\bibitem[Sims (2001)]{r31}
\textsc{Sims, C.A.}  (2001):
``Solving Linear Rational Expectations Models," \textit{Computational Economics}
\textbf{20} 1--20.
\bibitem[Uhlig (1999)]{r32}
\textsc{Uhlig, H.}  (1999):
``A Toolkit for Analysing Nonlinear Dynamic Stochastic Models Easily'' in:\textit{Computational Methods for the Study of Dynamic Economies} ed. by R. Marimon and A. Scott. 
Oxford, UK, Oxford University Press, 30-61
\end{thebibliography}

\end{document}